\documentclass{svproc}
\usepackage{url}

\usepackage[T1]{fontenc}
\usepackage{amssymb}
\usepackage{amsfonts}
\newtheorem{observation}{Observation}
\newtheorem{procedure}{Procedure}
\usepackage{xspace}
\usepackage[noend]{algpseudocode}
\usepackage[export]{adjustbox}
\usepackage{algorithm}
\usepackage{wrapfig}
\usepackage{bm}
\usepackage[caption=false,font=footnotesize]{subfig}
\usepackage{array}
\usepackage{multirow}
\usepackage{cleveref}
\newcounter{theconditions}
\crefname{theconditions}{condition}{conditions}
\newcommand{\defcond}[1]{
  \phantomsection
  \refstepcounter{theconditions}
  \label{#1}
  \textbf{Condition~\arabic{theconditions}}
}
\usepackage{paralist} 
\usepackage{tikz,graphicx}
\usetikzlibrary{calc,patterns,decorations.pathmorphing,decorations.markings,matrix,fit,decorations.pathreplacing,arrows,arrows.meta,automata,positioning,shapes,chains,spy}
\usetikzlibrary{intersections,through,backgrounds}
\usetikzlibrary{shadows,fadings}
\usepackage{tikz-3dplot}
\tikzset{
  invisible/.style={opacity=0},
  only/.code args={<#1>#2}{\only<#1>{\pgfkeysalso{#2}}},
  alt/.code args={<#1>#2#3}{\alt<#1>{\pgfkeysalso{#2}}{\pgfkeysalso{#3}}},
  temporal/.code args={<#1>#2#3#4}{%
    \temporal<#1>{\pgfkeysalso{#2}}{\pgfkeysalso{#3}}{\pgfkeysalso{#4}}},
  point/.style={circle,inner sep=1.5pt,minimum size=1.5pt,draw,fill=#1},
  point/.default=red,
  big arrow/.style={
    decoration={markings,mark=at position 1 with {\arrow[scale=1.5,#1]{>}}},
    postaction={decorate},
    shorten >=0.4pt},
  big arrow/.default=black}
\tikzset{cross/.style={%
    cross out, draw=black, minimum size=2*(#1-\pgflinewidth), inner sep=0pt,%
    outer sep=0pt},%
  cross/.default={1pt}}
\def\centerarc[#1](#2)(#3:#4:#5)
{ \draw[#1] ($(#2)+({#5*cos(#3)},{#5*sin(#3)})$) arc (#3:#4:#5); }
\def\calA{{\cal A}}
\def\calB{{\cal B}}
\def\calC{{\cal C}}
\def\calE{{\cal E}}
\def\calF{{\cal F}}

\def\calP{{\cal P}}
\def\calQ{{\cal Q}}
\def\calS{{\cal S}}
\def\calR{{\cal R}}

\newcommand{\R}{\ensuremath{\mathbb{R}}}               
\newcommand{\Rx}[1]{\ensuremath{\R\rule{0.3mm}{0mm}^{#1}}}
\newcommand{\Rone}{\Rx{1}\xspace}
\newcommand{\Rtwo}{\Rx{2}\xspace}

\newcommand{\Rd}{\Rx{d}\xspace}
\newcommand{\PPx}[1]{\ensuremath{\mathbb{P} \rule{0.3mm}{0mm}^{#1}}}
\newcommand{\PPone}{\PPx{1}\xspace}
\newcommand{\PPtwo}{\PPx{2}\xspace}
\newcommand{\SOx}[1]{\ensuremath{\mathbb{S}\rule{0.3mm}{0mm}^{#1}}}
\newcommand{\SOone}{\SOx{1}\xspace}
\newcommand{\SOtwo}{\SOx{2}\xspace}
\newcommand{\HOx}[1]{\ensuremath{\mathbb{H}\rule{0.3mm}{0mm}^{#1}}}
\newcommand{\HOtwo}{\HOx{2}\xspace}

\newcommand{\HS}{\ensuremath{\mathbb{H}}}               
\newcommand{\CS}{\ensuremath{\mathbb{A}}}               
\newcommand{\openCS}{\CS}
\newcommand{\closedCS}{\bar{\CS}}
\newcommand{\openHS}{\HS}
\newcommand{\closedHS}{\bar{\HS}}
\newcommand{\boundaryHS}{\partial{\HS}}
\newcommand{\cgal}{\textsc{Cgal}}

\definecolor{eficolor}{rgb}{0.459, 0.459, 1.0}
\definecolor{tomcolor}{rgb}{0.8235, 0.4117, 0.1176}
\definecolor{todocolor}{rgb}{0.0235, 0.3117, 0.0176}
\definecolor{color1}{rgb}{1, 0.5, 0.2}
\definecolor{color2}{rgb}{0.5, 0.5, 1.}
\graphicspath{{./figs/}}

\usepackage{ifthen}
\renewcommand{\Comment}[2][.45\linewidth]{\leavevmode\hfill\makebox[#1][l]{\color{red}//~#2}}
\newcommand{\CommentB}[2][.45\linewidth]{\leavevmode\hfill\makebox[#1][l]{\color{red}//~#2}}
\begin{document}
\title{Optimized Synthesis of Snapping Fixtures\thanks{This work has been supported in part by the Israel Science Foundation (grant nos.~825/15,1736/19), by the Blavatnik Computer Science Research Fund, and by grants from Yandex and from Facebook.}}
\titlerunning{Snapping Fixtures}
\author{Tom Tsabar\inst{1} \and
  Efi Fogel\inst{1} \and
  Dan Halperin\inst{1}}
\authorrunning{T. Tsabar et al.}
%
\institute{%
  The Blavatnik school of Computer Science, Tel Aviv University, Israel\\
  \email{\{tomtsabar9@gmail.com,efifogel@gmail.com,danha@post.tau.ac.il\}}%
}
%
\maketitle              
\newlength{\intextsepSaved}\setlength\intextsepSaved{\intextsep}%
\newlength{\columnsepSaved}\setlength\columnsepSaved{\columnsep}%
\newsavebox{\measurebox}
\newlength{\beforeSectionVSpace}\setlength\beforeSectionVSpace{-7pt}
\newlength{\afterSectionVSpace}\setlength\afterSectionVSpace{-7pt}
\newlength{\subsectionVSpace}\setlength\subsectionVSpace{-3pt}
\newlength{\abovecaptionskipSaved}
\newlength{\belowcaptionskipSaved}
\setlength{\abovecaptionskipSaved}{\abovecaptionskip}
\setlength{\belowcaptionskipSaved}{\belowcaptionskip}
\begin{abstract}
  Fixtures for constraining the movement of parts have been
  extensively investigated in robotics, since they are essential for
  using robots in automated manufacturing.
  This paper deals with the design and optimized synthesis of a special
  type of fixtures, which we call \emph{snapping fixtures}. Given a
  polyhedral workpiece $P$ with $n$ vertices and of constant genus,
  which we need to hold, a snapping fixture is a semi-rigid polyhedron
  $G$, made of a palm and several fingers, such that when $P$ and $G$
  are well separated, we can push $P$ toward $G$, slightly bending the
  fingers of $G$ on the way (exploiting its mild flexibility), and
  obtain a configuration, where $G$ is back in its original shape and
  $P$ and $G$ are inseparable as rigid bodies. We prove the minimal
  closure conditions under which such fixtures can hold parts, using
  Helly's theorem. We then introduce an algorithm running in $O(n^3)$
  time that produces a snapping fixture, minimizing the number
  of fingers and optimizing additional objectives, if a snapping fixture
  exists. We also provide an efficient and robust implementation of a
  simpler version of the algorithm, which produces the fixture model
  to be 3D printed and runs in $O(n^4)$ time. We describe two
  applications with different optimization criteria: Fixtures to hold
  add-ons for drones, where we aim to make the fixture as lightweight as
  possible, and small-scale fixtures to hold precious stones in jewelry,
  where we aim to maximize the exposure of the stones, namely minimize
  the obscuring of the workpiece by the fixture.
\end{abstract}

\section{Introduction}
\label{sec:intro}
A fixture is a device that holds a part in place. Constraining the
movement of parts is a fundamental requirement for using robots in
automated manufacturing~\cite{a-mrbic-96},\cite[Section~3.5]{w-irs-15}.
There are many types
and forms of fixtures; they range from modular fixtures synthesized on
a lattice to fixtures generated to suit a specific part. A fixture
possesses some grasp characteristics. For example, a grasp with complete
restraint prevents loss of contact, prevents any motion, and thus may
by considered secure. Two primary kinematic restraint properties are
\emph{form closure} and \emph{force closure}~\cite{pt-g-16}.
Both properties guarantee maintenance of contact under some conditions.
However, the latter typically relies on contact friction; therefore,
achieving force closure typically requires fewer contacts than
achieving form closure.  Fixtures with complete restraint are mainly
used in manufacturing processes where preventing any motion is
critical. Other types of fixtures can be found anywhere, for example,
in the kitchen where a hook holds a cooking pan, or in the office
where a pin and a bulletin board hold a paper still. This paper deals
with a specific problem in this area; here, we are given a rigid
object, referred to as the \emph{workpiece}, and we seek for an
automated process that designs a semi-rigid object, referred to as the
snapping \emph{fixture}, such that, starting at a configuration where
the workpiece and the holding fixture are separated, they can be
pushed towards each other, applying a linear force and exploiting the
mild flexibility of the fixture, into a configuration where both the
workpiece and the fixture are inseparable as rigid bodies. A generated
fixture has a base part, referred to as the \emph{palm}, and fingers
connected to the palm; see Section~\ref{ssec:terms:structure} for
formal definitions. Without additional computational effort, a hook, a
nut, or a bolt can be added to the palm resulting in a generic fixture
that can be utilized in a larger system.  Another advantage of the
single-component flexible fixture is that it can easily be 3D-printed.
We have 3D-printed several fixtures that our generator has
automatically synthesized for some given workpieces.  The objective of
the algorithm is obtaining snapping fixtures with the minimal number
of fingers.  With additional care that also accounts for properties of
the material used to produce the fixtures, the smallest or lightest
possible fixture can be synthesized, for a given workpiece. This can
(i) expedite the production of the fixture using, e.g., additive
manufacturing, (ii) minimize the weight of the produced fixture, and
(iii) maximize the exposed area of the boundary of the workpiece when
held by the fixture.


\subsection{Background}
\label{ssec:intro:background}
Form closure has been studied since the 19th century. Early results
showed that at least four frictionless contacts are necessary for
grasping an object in the plane, and seven in 3D space. Specifically,
it has been shown that four and seven contacts are necessary and
sufficient for the form-closure grasp of any polyhedron in the 2D and
3D case, respectively~\cite{mnp-gg-90,mss-esmpg-87}.

Automatic generation of various types of fixtures, and in particular,
the synthesis of form-closure grasps, are the subjects of a diverse
body of research. Brost and Goldberg~\cite{bg-casmf-94} proposed
a complete algorithm for synthesizing modular fixtures of polygonal
workpieces by locating three pegs (locators), and one clamp on a
lattice. Their algorithm is complete in the sense that it examines all
possible fixtures for an input polygon. Their results were obtained by
generating all configurations of three locators coincident to three
edges, for each triplet of edges in the input polygon. For each such
configuration, the algorithm checks whether \emph{form closure} can be
obtained by adding a single clamp. Our work uses a similar
strategy to obtain all possible configurations. In subsequent work
Zhuang, Goldberg, and Wong~\cite{zk-esmf-96} showed that there exists
a non-trivial class of polygonal workpieces that cannot be held in
form closure by any fixture of this type (namely, a fixture that uses
three locators and a clamp). They also considered fixtures that use
four clamps, and introduced two classes of polygonal workpieces that
are guaranteed to be held in form closure by some fixture of this
type. Wallack
and Canny~\cite{wc-pmhf-97} proposed another type of fixture called
the vise fixture and an algorithm for automatically designing such
fixtures. The vise fixture includes two lattice plates mounted on the
jaws of a vise and pegs mounted on the plates. Then, the workpiece is
placed on the plates, and \emph{form closure} is achieved by
activating the vise and closing the pins from both sides on the
workpiece. The main advantage in this type of fixture is its
simplicity of usage. Brost and Peters~\cite{bp-adfap-98} extended
the approach exploited in~\cite{bg-casmf-94} to three dimensions. They
provided an algorithm that generates suitable fixtures for
three-dimensional workpieces.  Wagner, Zhuang, and
Goldberg~\cite{wzg-ffpsm-95} proposed a three-dimensional
seven-contact fixture device and an algorithm for planning
form-closure fixtures of a polyhedral workpiece with pre-specified
pose. A summary of the studies in the field of flexible fixture design
and automation conducted in the last century can be found
in~\cite{bz-ffdar-01}. Related studies in the field of grasping and manipulation are summarized in~\cite{bk-rgm-01}.
Subsequent works studied other
types of fixtures and provided algorithms for computing them, for
example, unilateral fixtures~\cite{ggbzk-ufsmp-04}, which are used to
fix sheet-metal workpieces with holes. Other studies focused on
grasping synthesis algorithms with autonomous robotic \emph{fingers},
where a single robotic hand gets manipulated by motors and being used
to grasp different workpieces; an overview of such algorithms can be
found in ~\cite{seb-oogsa-12}. A common dilemma for all the grasping
and fixture design algorithms is defining and finding the optimal
grasp. Several works, e.g.,~\cite{lwxl-qfgsf-04}
and~\cite{rpr-goscc-13}, discuss such quality functions and their
optimization.

\subsection{Our Results}
\label{ssec:intro:results}
We introduce certain properties of minimal snapping fixtures of given
workpieces. Formally, we are given a closed polyhedron $P$ of
complexity $n$ and of a constant genus that represents a workpiece.
The surface of a polyhedron of genus zero is homeomorphic to a
sphere. In our work we allow more complicated polyhedra; see, for
example, Figure~\ref{fig:ms:a}.\footnote{The genus counts the number of
``handles'' in the polyhedron; see, e.g.,
\url{https://mathworld.wolfram.com/Genus.html}.} In our analysis
in the sequel we assume that the genus is bounded by a constant.
We introduce an algorithm that determines whether a closed polyhedron $G$
that represents a fixture exists, and if so, it constructs it in
$O(n^3)$ time. This significantly improves our simpler $O(n^4)$
algorithm listed in Appendix~\ref{sec:simple-algorithm}. We also provide
an efficient and
robust implementation of the latter. In addition, we present two
practical cases that utilize our implemented algorithm: One is the
generation of a snapping fixture that mounts a device to an unmanned
aerial vehicle (UAV), such as a drone. The other is the generation of
a snapping fixture that mounts a precious stone to a jewel, such as a
ring. The common objective in both cases is, naturally, the firm
holding of the workpiece. In the first case, we are interested in a
fixture with minimal weight. In the second case we are interested in a
fixture that minimally obscures the precious stone.
We are not aware of similar work on semi-rigid one-part fixtures; thus,
we do not conduct any comparisons, but we provide some benchmark
numbers we have obtained while executing our generator. Note
that, in theory, the generated fixtures prevent any linear motion, but
do not necessarily prevent angular motion;
however, fixtures that do not posses the \emph{form closure} property
are rarely obtained in practice. Handling angular motion is left for
future research; see Appendix~\ref{ssec:future:form-closure}.

\subsection{Outline}
\label{ssec:intro:outline}
The rest of this paper is organized as follows. Terms and definitions
for our snapping fixtures and theoretical bounds and properties are
provided in Section~\ref{sec:terms-properties}. The synthesis
algorithm is described in Section~\ref{sec:algorithm} along with the
analysis of its complexity. Two applications are presented in
Section~\ref{sec:cases}. Finally, we report on experimental results in
Section~\ref{sec:experiments}.
Appendix~\ref{sec:glossary} provides a notation glossary.
Appendix~\ref{sec:proofs} contains proofs of several lemmas,
corollaries, observations, and a theorem. We point out
some limitations of our generator and suggest future research in
Appendix~\ref{sec:future}.  An assortment of interesting workpieces
and fixtures that we have 3D-printed and experimented with are shown
in Appendix~\ref{sec:assortment}.

\section{Terminology and Properties}
\label{sec:terms-properties}
In this section we describe the structure and properties of our
snapping fixtures.

\subsection{Fixture Structure}
\label{ssec:terms:structure}
\setlength{\intextsep}{0pt}%
\setlength{\columnsep}{4pt}%
\begin{wrapfigure}[5]{R}{4.0cm}
  \includegraphics[width=4.0cm]{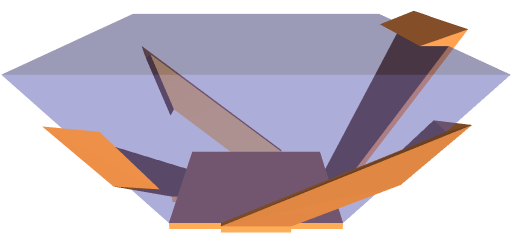}
\end{wrapfigure}
Consider an input polyhedron $P$ that represents a workpiece, such as
the one transparently rendered in blue in the figure to the right. The
structure of a fixture of $P$, rendered in orange in the figure,
resembles the structure of a hand; it is the union of a single
polyhedral part referred to as the \emph{palm}, several polyhedral
parts, referred to as \emph{fingers}, which are extensions of the
\emph{palm}, and semi-rigid joints that connect the palm and the
fingers. Each \emph{finger} consists of two polyhedral parts, namely,
\emph{body} and \emph{fingertip}, and the semi-rigid joint between the
\emph{body} and the \emph{fingertip}. The various parts, i.e., palm,
bodies, and fingertips, are disjoint in their interiors.
In the following we describe these parts in detail.
\setlength\intextsep{\intextsepSaved}%
\setlength\columnsep{\columnsepSaved}%

\begin{definition}[$\bm{\alpha}$-extrusion of a polygon and base polygon of an
  $\alpha$-extrusion] Let $L$ denote a polygon in space, let $v$
  denote a normal to the plane containing $L$, and let $v_{\alpha}$
  denote the normal scaled to length $\alpha$. The
  \emph{$\alpha$-extrusion} of $L$ is a polyhedron $Q$ in space,
  which is the extrusion of $L$ along $v_{\alpha}$. The polygon $L$ is
  referred to as the \emph{base polygon} of $Q$; see the figure below.
\end{definition}

\setlength{\intextsep}{-5pt}%
\setlength{\columnsep}{0pt}%
\begin{wrapfigure}[7]{R}{2.45cm}
  \begin{tikzpicture}[node distance=5cm]
    \node (img) {\includegraphics[width=\linewidth]{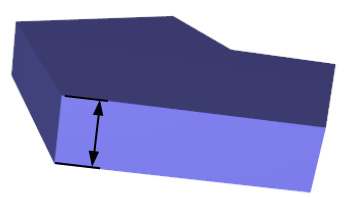}};
    \node [yshift=-7, xshift=-6] at (img.center) {\large $\alpha$};
  \end{tikzpicture}\vspace{-4pt}\\
  \begin{tikzpicture}[node distance=5cm]
    \node (img) {\includegraphics[width=\linewidth]{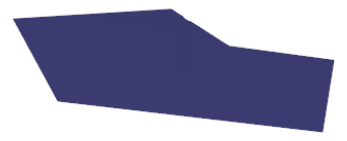}};
    \draw [thick,-{Stealth[scale=1.5]}] (0.6,-0.1)--(0.7,0.5);
    \node [above,yshift=-6, xshift=20] at (img.north){\large $v$};
  \end{tikzpicture}
\end{wrapfigure}
We use the abbreviation $\alpha$-extrusion of a facet
$f$ of some polyhedron to refer to the $\alpha$-extrusion $Q$ of the
geometric embedding of the facet $f$, and we refer to the facet of $Q$
that overlaps with $f$ as the base facet of the $\alpha$-extrusion $Q$.
\setlength\intextsep{\intextsepSaved}%
\setlength\columnsep{\columnsepSaved}%

Our formal computational model is oblivious to the thickness of the
various parts. In this model the parts are flat and if two parts are
connected by a joint, they share an edge, which is the axis of the
joint. Our generator, though, synthesizes solid models of fixtures. We
use $\alpha$-extrusion to inflate the various parts.

Let $G$ denote a snapping fixture made of a palm, $k$ fingers
$F_1,F_2,\ldots,F_k$, and corresponding joints. The palm is an
$\alpha_p$-extrusion of a facet $f_p$ of the workpiece $P$.
(The various $\alpha$ values are discussed below.)
Consider a specific finger $F=F_i$
of $G$. The body of $F$ is defined by one of the neighboring facets of
$f_p$, denoted $f_b$. The fingertip of $F$ is defined by one of the
neighboring facets of $f_b$, denoted $f_t$, $f_t\neq f_p$. Let
$e_{pb}$ denote the common edge of $f_p$ and $f_b$, and let $e_{bt}$
denote the common edge of $f_b$ and $f_t$. Note that in some
degenerate cases $e_{pb}$ and $e_{bt}$ are incident to a common
vertex. The body of a finger is an $\alpha_b$-extrusion of $f_b$.
Let $v$ denote the cross
product of the vector that corresponds to $e_{bt}$ and the normal to
the plane containing $f_t$ of length $\alpha_t$.
Let $q_t$ denote the quadrilateral defined by the two vertices incident to
$e_{bt}$ and their translations by $v$. The fingertip is an
$\alpha_t$-extrusion of $q_t$. The axis of the joint that connects the
palm and the body of $F$ coincides with $e_{pb}$ and the axis of the
joint that connects the body of $F$ with its fingertip coincides with
$e_{bt}$.  The value $\alpha_p$ and the values $\alpha_b$ and
$\alpha_t$ for each finger determine the trade-off between the
strength and flexibility of the joints.\footnote{Typically, these
values are identical.} They depend on the material and shape of the
fixture. In our implementation they can be determined by the
user.\footnote{For example, in several of the fixtures that we
produced, they were set to 5mm.}

For a complete view of a workpiece and a snapping fixture consider
Figure~\ref{fig:wrap}. Observe that both the palm and the fingers of
the fixture in the figure differ from the formal definitions above.
The differences stem from practical considerations. In particular, the
parts in the figure have smaller volumes, which (i) reduces
fabrication costs, and (ii) resolves collision between distinct
fingers. In some degenerate cases (see Figure~\ref{fig:degenerate})
distinct fingers could have overlapped.  In the figure, the base facet
of the fingertip of one finger, $f_{t_1}$, coincides with $f$, a facet
of the workpiece.  Likewise, the base facet of the body of the other
finger, $f_{b_2}$, also coincides with $f$. Avoiding overlaps is
achieved by simultaneously shrinking the base facets $f_{t_1}$ and
$f_{b_2}$. Now, the fingertip grips only the tip of $f$ and the body is
stretching only on a small portion of the workpiece facet. As another
example, consider the body of a finger depicted in
Figure~\ref{fig:wrap}\subref{fig:fix1:t}; it is the
$\alpha_b$-extrusion of a quadrilateral defined by two points that lie
in the interior of $e_{pb}$ and two points that lie in the interior of
$e_{bt}$, as opposed to the formal definition above, where the body is
the $\alpha_b$-extrusion of the entire facet of $P$.  Also, in
reality, parts are not fabricated separately, and the entire fixture
is made of the same flexible material. Instead of rotating about the
joint axes, the entire fingers bend. The differences, though, have
no effect on the correctness of the proofs and algorithm (which adhere
to the formal definitions) presented in the sequel. These structural
changes and the extrusion values, merely determine the degree of
flexibility and strength of the fixture; see
Appendix~\ref{ssec:future:flexibility} and
Appendix~\ref{ssec:future:strength}.

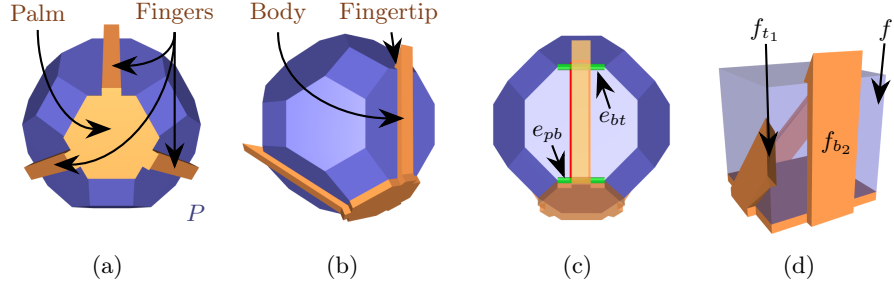
\begin{figure}[!ht]
  \vspace{-25pt}
  \centering\subfloat[]{\label{fig:fix1:o}
    \begin{tikzpicture}
      \node (img) {\includegraphics[height=2.8cm]{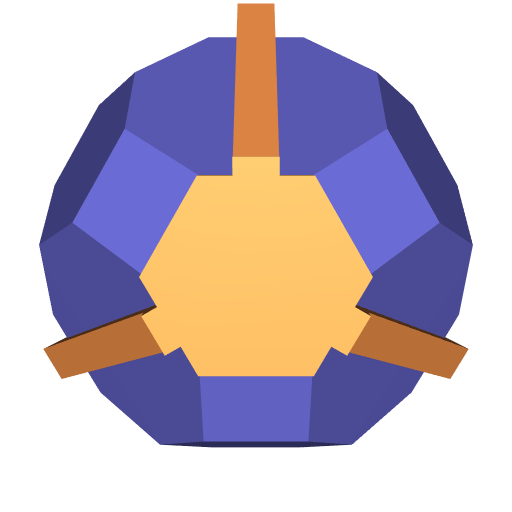}};
      \begin{scope}[scale=1.4]
        \node [](P) at (0.8,-0.8) {\textcolor{color2!50!black}{$P$}};
        \node [](ps) at (-0.7, 1.1) {\textcolor{color1!50!black}{Palm}};
        \node [](fs) at (0.6, 1.1) {\textcolor{color1!50!black}{Fingers}};
        \coordinate (pe) at (0.,0.);
        \draw [thick,-{Stealth[scale=1.5]}] (ps) to[out=-90,in=180](pe);
        \coordinate (fe1) at (0.6,-0.38);
        \coordinate (fe2) at (0,0.6);
        \coordinate (fe3) at (-0.55,-0.34);
        \draw [thick,-{Stealth[scale=1.5]}] (fs)--(fe1);
        \draw [thick,-{Stealth[scale=1.5]}] (fs) to[out=-90,in=0](fe2);
        \draw [thick,-{Stealth[scale=1.5]}] (fs) to[out=-90,in=0](fe3);
      \end{scope}
  \end{tikzpicture}}
  \subfloat[]{\label{fig:fix1:f}
    \begin{tikzpicture}[node distance=10cm]
      \node (img) {\includegraphics[height=2.8cm]{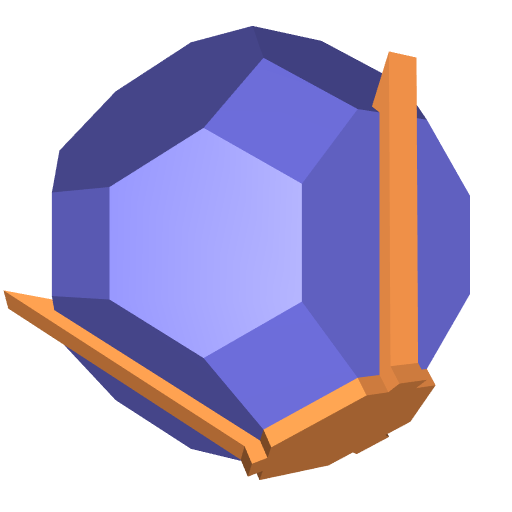}};
      \begin{scope}[scale=1.4]
        %
        \node [](ts) at (0.4, 1.1) {\textcolor{color1!50!black}{Fingertip}};
        \node [](bs) at (-0.65, 1.1) {\textcolor{color1!50!black}{Body}};
        \coordinate (te) at (0.47,0.61);
        \coordinate (be) at (0.55,0.1);
        \draw [thick,-{Stealth[scale=1.5]}] (ts)to(te);
        \draw [thick,-{Stealth[scale=1.5]}] (bs)to[out=-90,in=180](be);
      \end{scope}
  \end{tikzpicture}}
  \subfloat[]{\label{fig:fix1:t}
    \begin{tikzpicture}
      \node (img) {\includegraphics[height=2.8cm]{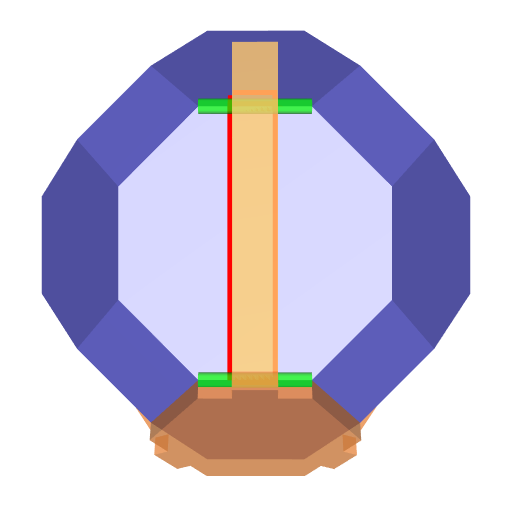}};
      \begin{scope}[scale=1.4]
        \node [inner sep=1pt](epbs) at (-0.32, -0.05) {\textcolor{black}{$e_{pb}$}};
        \coordinate (epbe) at (-0.16,-0.47);
        \draw [thick,-{Stealth[scale=1.5]}] (epbs)--(epbe);
        \node [,inner sep=1pt](ebgs) at (0.3, 0.1) {\textcolor{black}{$e_{bt}$}};
        \coordinate (ebte) at (0.15,0.55);
        \draw [thick,-{Stealth[scale=1.5]}] (ebgs)--(ebte);
      \end{scope}
  \end{tikzpicture}}
  \subfloat[]{\label{fig:degenerate}
    \begin{tikzpicture}[node distance=5cm]
      \node (img) {\includegraphics[width=0.20\linewidth]{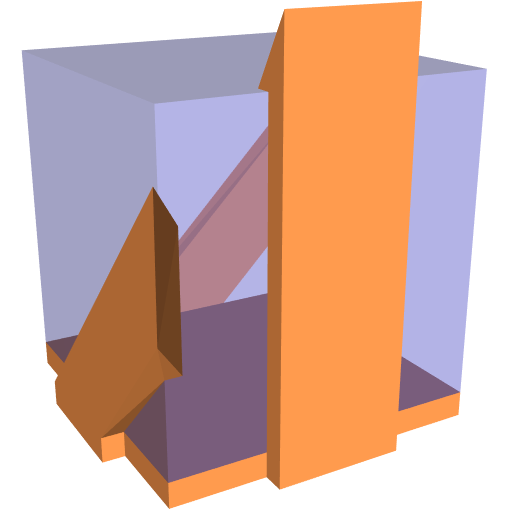}};
      \node [inner sep=1pt](ts) at (-0.5, 1.5) {$f_{t_1}$};
      \node [inner sep=1pt](bs) at (0.45, 0.0) {$f_{b_2}$};
      \node [inner sep=1pt](fs) at (1.1, 1.5) {$f$};
      \coordinate (te) at (-0.45,-0.2);
      \coordinate (be) at (0.1,0.0);
      \coordinate (fe) at (1.0,0.5);
      \draw [thick,-{Stealth[scale=1.5]}] (ts)--(te);
      \draw [thick,-{Stealth[scale=1.5]}] (fs)--(fe);
    \end{tikzpicture}}
  \setlength{\abovecaptionskip}{0pt}
  \setlength{\belowcaptionskip}{0pt}
  \caption[]{(\subref*{fig:fix1:o}), (\subref*{fig:fix1:f}),
    (\subref*{fig:fix1:t}) Different views of a truncated cuboctahedron
    (blue) and a snapping fixture (orange).
    (\subref*{fig:degenerate}) A transparent cube (blue) and a snapping
    fixture (orange).}
  \label{fig:wrap}
  \vspace{-20pt}
\end{figure}

\subsection{The Configuration Space}
\label{ssec:terms:configuration-space}
The workpiece and its snapping fixture form an assembly.  Each joint in
the fixture connects two parts; it enables the rotation of one part
with respect to the other about an axis. Each joint adds one degree
of freedom (DOF) to the configuration space of the assembly.

In our context, the workpiece and its snapping fixture are considered
assembled, if they are \emph{infinitesimally inseparable}. When two
polyhedra are infinitesimally inseparable, any linear motion applied
to one of the polyhedra causes a collision between the polyhedra interiors.
  The workpiece and the fixture are in the \emph{serving configuration}
  if (i) they are separated (that is, they are arbitrarily far away
  from each other), and (ii) there exists a vector $v$, such that when
  the fixture is translated by $v$, as a result of some force applied
  in the direction of $v$, exploiting the flexibility of the joints of
  the fixture, the workpiece and the fixture become assembled.
When the workpiece and its snapping fixture are separated, the fixture
can be transformed without colliding with the workpiece to reach the
serving configuration.\footnote{The video clip available at
\url{http://acg.cs.tau.ac.il/projects/ossf/snapping_fixtures.mp4}
illustrates the snapping operation.}

\subsection{Spreading Degree}
\label{sssec:terms:spreading-degree}
\vspace{\subsectionVSpace}
\setlength{\intextsep}{0pt}%
\begin{wrapfigure}[4]{r}{1.45cm}
  \centering%
  \includegraphics[trim={1cm 0 0 0},width=1.45cm]{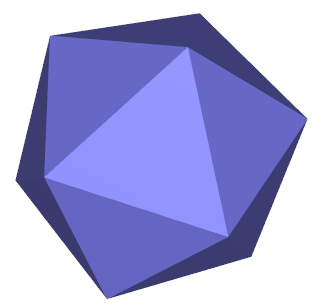}
\end{wrapfigure}
The \emph{spreading degree} is the number of facets involved in the
definition of a finger.  In this paper we restrict ourselves to
snapping fixtures that have fingers with spreading degree two, which
means that the body of every finger is based on a single facet of $P$.
Every finger (the body and the fingertip) stretches over two facets of
$P$. Naturally, fingers with a higher spreading-degree reach
further. An icosahedron, for example, (depicted in the figure above)
does not admit a valid fixture with spreading degree two. This is
proven by exhaustion running our implemented algorithm.
\setlength\intextsep{\intextsepSaved}%
\setlength\columnsep{\columnsepSaved}%

\subsection{Fixture Planning}
\label{ssec:terms:planning}
\vspace{\subsectionVSpace}
The basic objective of our fixture algorithms is obtaining fixtures
with the minimal number of fingers. Our generator is of the exhaustive
type. As explained in Section~\ref{sec:algorithm}, it examines many
different possible candidates of fingers, before it reaches a
conclusion. The simple (and implemented) algorithm, for example,
visits every valid fixture (of 2,~3, or~4) fingers; thus, it can be
used to produce all or some valid fixtures according to any
combination of optimization criteria.  As aforementioned, the
generator synthesizes fixtures of spreading degree two. Extending the
generator to enable the synthesis of fixtures with an increased
spreading degree (without further modifications) will directly
increase the search space exponentially.

\subsection{Properties}
\label{ssec:terms:properties}
\vspace{-2pt}

\begin{definition}[unit circle, semicircle, open semicircle]
  An open semicircle is a semicircle excluding its two endpoints. An
  open hemisphere is a hemisphere excluding the great circle that
  comprises its boundary curve.
\end{definition}

\begin{definition}[Covering set]
\looseness=-1%
  Let $\calS=\{s_1,...,s_{|\calS|}\}$ be a finite set of subsets of
  $\Rd$ and $C$ be a set of points in $\Rd$. If
  $\bigcup_{i=1}^{|\calS|} s_{i} \supseteq C$ then $\calS$ is a
  covering set of $C$.
\end{definition}


A pair of open unit semicircles (respectively, hemispheres) are called
antipodal if the closure of their union is the entire unit circle
(respectively, sphere).


For lack of space, we defer portions of the formal analysis to the
appendix. In particular a sequence of lemmas proved in the appendix
yield the following corollaries. The proofs of the corollaries,
observations and theorem in the remainder of this section, appear in
Appendix~\ref{sec:proofs}.

\begin{corollary}\label{corollary:parallelepiped:2:minimal}
  Let $\calR$ be a set of four open unit semicircles that cover the unit
  circle $\SOone$. $\calR$ is minimal (i.e., for every open semicircle
  $s \in \calR$, $\calR \setminus \{s\}$ is not a covering set of \SOone)
   iff it consists of two antipodal pairs of open unit semicircles.
\end{corollary}

\begin{corollary}\label{corollary:parallelepiped:2:no}
  Let $\calS$ be a set of distinct open unit semicircles that covers
  $\SOone$; if $|\calS| \geq 5$, then there exists $\calR \subset
  \calS$, $|\calR|=3$ and $\calR$ covers $\SOone$.
\end{corollary}

Generalizing Corollaries~\ref{corollary:parallelepiped:2:minimal}
and~\ref{corollary:parallelepiped:2:no} to 3-space yields the
following.

\begin{corollary}\label{corollary:parallelepiped:3}
  Let $\calR$ be a set of six open unit hemispheres that cover the unit
  sphere $\SOtwo$. $\calR$ is minimal iff it consist of three antipodal
  pairs of open unit hemispheres.
\end{corollary}


\begin{corollary}\label{corollary:parallelepiped:3:no}
  Let $\calS$ be a set of distinct open unit hemispheres that covers
  $\SOtwo$; if $|\calS| \geq 7$, then there exists $\calR \subset
  \calS$, $|\calR|=5$ and $\calR$ covers $\SOtwo$.
\end{corollary}

When a facet $f$ of the workpiece partially coincides with a facet of the
fixture, the workpiece cannot translate in any direction that forms an
acute angle with the (outer) normal to the plane containing $f$
(without colliding with the fixture). This set of blocking directions
comprises an open unit hemisphere denoted as $h(f)$. Similarly,
$H(\calF) = \{h(f)\,|\,f \in \calF\}$ denotes the mapping from a
set of facets to the set of corresponding open unit hemispheres; see,
e.g., \cite{shb-spspm-17}. Let $\calF'$ denote the set of facets of
the workpiece that are coincident with facets of the fixture in some
fixed configuration. If the union of all blocking directions covers the
unit sphere in that configuration, formally stated $\SOtwo = \bigcup
H(\calF')$, then the workpiece cannot translate at all.

Let $\calF$ denote the set of all facets of the fixture $G$. Let
$\calF_{P}$ denote the singleton that consists of the base facet of
the palm of $G$, and let $f_{b_{i}}$ and $f_{t_{i}}$, $1\leq i\leq k$,
denote the base facet of the body and the base facet of the fingertip,
respectively, of the $i$-th finger of $G$, where $k$ indicates the
number of fingers. Let $\calF_{B}=\{f_{b_{i}}\,|\,1\leq i\leq k\}$ and
$\calF_{T}=\{f_{t_{i}}\,|\,1\leq i\leq k\}$ denote the set of the base
facets of the body parts of the fingers of $G$ and the set of the base
facets of the fingertip parts of the fingers of $G$, respectively. Let
$\calF_{PBT}$ denote the set of all base facets of the parts of $G$,
that is $\calF_{PBT} = \calF_{P} \cup \calF_{B} \cup \calF_{T}$.  Let
$\calF_{PB}$ denote the set of all base facets of the parts of $G$
excluding the base facets of the fingertips, that is, $\calF_{PB} =
\calF_{P} \cup {\calF}_{B}$.

If the fixture resists any linear force applied on the workpiece while
in the assembled state and there exists a collision free path (in the
configuration space) between any separated configuration and the
assembled configuration then our fixture is valid. We relax the second
condition for practical reasons; instead of requiring a full path,
we require a path of infinitesimal length. Formally we get:

\begin{enumerate}
\item[\defcond{expression:valid-fixture:1}:] $\SOtwo=\bigcup H(\calF_{PBT})$.
\item[\defcond{expression:valid-fixture:2}:] $\SOtwo\neq\bigcup H(\calF_{PB})$.
\end{enumerate}
If the second condition holds, a serving state exists (assuming the
flexibility of the joints cancels out the obstruction induced by the
presence of the fingertips).

A candidate finger of an input polyhedron $P$ is a valid finger of at
least one possible fixture $G$ of $P$.

\begin{observation}\label{observation:candidate-fingers}
  The number of candidate fingers of an input polyhedron $P$ is linear
  in the number of vertices of $P$.
\end{observation}

\begin{theorem}\label{theorem:upper-bound}
  Every valid snapping fixture can be converted to a four-finger snapping fixture.
  Sometimes four fingers are necessary.
\end{theorem}

\begin{observation}\label{observation:one-finger}
  A single-finger fixture does not exist.
\end{observation}

\begin{figure}[ht]
  \centering%
  \subfloat[][]{\label{fig:special:1}%
    \includegraphics[height=2.6cm]{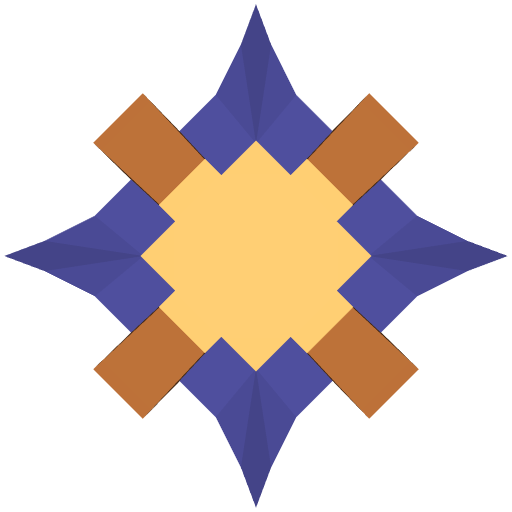}}\quad
  \subfloat[][]{\label{fig:special:2}%
    \includegraphics[height=2.6cm]{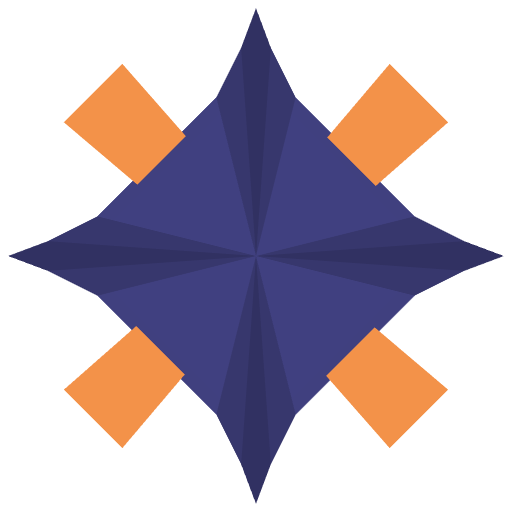}}\quad
  \subfloat[][]{\label{fig:special:3}%
    \includegraphics[height=2.6cm]{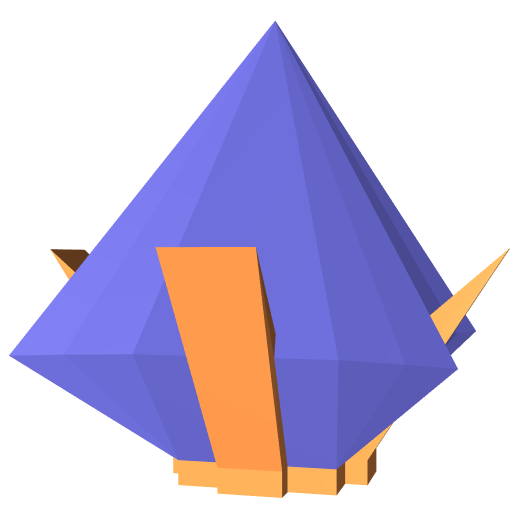}}\quad\quad\quad\quad
  \subfloat[][]{\label{fig:two-fingers}%
    \includegraphics[height=2.6cm]{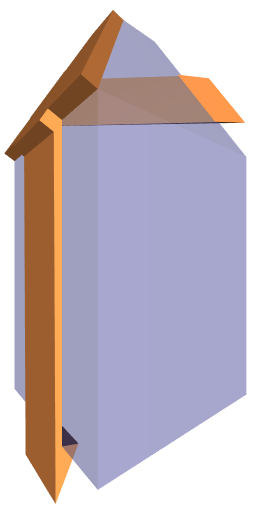}}
  \setlength{\abovecaptionskip}{0pt}
  \setlength{\belowcaptionskip}{-10pt}
  \caption[]{%
    (\subref*{fig:special:1}), (\subref*{fig:special:2}),
    (\subref*{fig:special:3})
    Different views of a polyhedron that has snapping fixtures with four
    fingers only and one of its four-finger fixtures.
    (\subref*{fig:two-fingers}) A snapping fixture with two fingers.
  }%
  \label{fig:special}
\end{figure}

A polyhedron that admits the lower bound is depicted in
Figures~\ref{fig:special:1}, \ref{fig:special:2},
and \ref{fig:special:3}. There there exists a polyhedron that has a
snapping fixture that has only two fingers; see
Figure~\ref{fig:two-fingers}.

\vspace{-7pt}
\section{Algorithm}
\label{sec:algorithm}
\vspace{-7pt}
A snapping fixture $G$ (of spreading degree two) is formally defined
by a pair that consists of
\begin{inparaenum}[(i)]
\item an index $i$ of a facet of $P$, and
\item a set of pairs of indices
  $(j_1,\ell_1),(j_2,\ell_2),...(j_k,\ell_k)$ of facets of $P$.
\end{inparaenum}
The palm of $G$ is the $\alpha_p$-extrusion of the facet $f_i$. Each
member pair of indices $(j, \ell)$ define a finger of $G$. The body
and fingertip of the finger are the $\alpha_b$- and $\alpha_t$-extrusion
values of the facets $f_j$ and $f_{\ell}$, respectively.

A simple algorithm that exhaustively searches through all valid snapping
fixtures with~2,~3, or~4, fingers, of a given polyhedron and runs in
$O(n^4)$ time is listed in Appendix~\ref{sec:simple-algorithm}. Here,
we introduce a much more parsimonious algorithm that uses a different
method to generate 4-finger fixtures, yielding an algorithm that
generates one fixture if exists with the minimal number of fingers and
runs in $O(n^3)$ time.


\begin{procedure}\normalfont%
  (\textbf{\textsc{minimalSnappingFixture}}$(P)$)
  The procedure accepts a polyhedron $P$ as input and returns a
  fixture of $P$ if exists with the minimal number of fingers; see
  Algorithm~\ref{alg:snapfixt}.  The algorithm consists of two phases.
  In the first phase we compute a data structure $M$ that associates
  palms and candidate fingers that extend from them. The second phase
  consists of three subphases in which we extract subsets of fingers
  of size,~2,~3, and~4, respectively, for each palm stored in $M$ and
  examine whether the palm and the subset of fingers form a valid
  fixture.  Once we strike one, we return it.
\end{procedure}

\begin{algorithm}[htb]
  \def\phase{\rlap{\smash{$\left.\begin{array}{@{}c@{}}\\{}\\{}\\{}\\{}\end{array}\color{red}\right\}%
        \color{red}\begin{tabular}{l}Phase 1\end{tabular}$}}}
  \def\subphase#1{\rlap{\smash{$\left.\begin{array}{@{}c@{}}\\{}\\{}\\{}\end{array}\color{red}\hspace*{-70pt}\right\}%
        \color{red}\begin{tabular}{l}Subphase #1\end{tabular}$}}}
  \caption{Minimal snapping fixture generation}
  \label{alg:snapfixt}
  \begin{algorithmic}[1]
    \Require A polyhedron $P$ with $m$ facets $\{f_1,f_2,...,f_m\}$.
    \Ensure A snapping fixture $G$ of $P$, if exists, with minimal
    number of fingers.
    \Procedure{minimalSnappingFixture}{$P$}
      \For{$i\gets 1, m$}
        \State{$M[i]\gets \emptyset$}
	\ForAll{$j, f_j \in \Call{neighbors}{f_i}$}\hspace*{100pt}\phase
	  \ForAll{$\ell, f_{\ell} \in \Call{neighbors}{f_j}\,\&\,\ell \neq i$}
	  \State{$M[i]\gets M[i]\cup \{(j,\ell)\}$}
	  \EndFor
	\EndFor
      \EndFor
      \For{$i\gets 1, m$}
        \ForAll{$\calS, \calS \in \Call{subsets}{M[i], 2}$}\Comment{$|\calS| = 2$}
        \State{$F\gets (f_i,\calS)$}\CommentB{Define a fixture}%
          \raisebox{0.5\baselineskip}{\subphase{2.1}}
          \If{$\Call{validFixture}{F}$} \Return{$F$}%
          \EndIf
        \EndFor
      \EndFor
      \For{$i\gets 1, m$}
        \ForAll{$\calS, \calS \in \Call{subsets}{M[i], 3}$}\Comment{$|\calS| = 3$}
        \State{$F\gets (f_i,\calS)$}\CommentB{Define a fixture}%
          \raisebox{0.5\baselineskip}{\subphase{2.2}}
          \If{$\Call{validFixture}{F}$} \Return{$F$}%
          \EndIf
        \EndFor
      \EndFor
      \For{$i\gets 1, m$}
        \State $F \gets\Call{fourFingersFixture}{f_i, M[i]}$
        \If{$F \neq \textrm{null}$} \Return{$F$}%
          \hspace*{191pt}\raisebox{0.5\baselineskip}{\subphase{2.3}}
        \EndIf
      \EndFor
      \State \Return{null}
    \EndProcedure
  \end{algorithmic}
\end{algorithm}

\begin{procedure}\normalfont%
  (\textbf{\textsc{neighbors}}$(f)$)
  The procedure accepts a facet $f$ of a polyhedron and returns all the
  neighboring facets of $f$.
\end{procedure}

\begin{procedure}\normalfont%
  (\textbf{\textsc{subsets}}$(\calC, k)$)
  The procedure accepts a set $\calC$ and a positive integer $k$; it
  returns all subsets of $\calC$ of cardinality $k$.
\end{procedure}

\begin{procedure}\normalfont%
  (\textbf{\textsc{validFixture}}$(F)$) The procedure accepts a
  snapping fixture and determines whether it is a valid snapping
  fixture based
  on~\Cref{expression:valid-fixture:1,expression:valid-fixture:2}
  defined in Section~\ref{ssec:terms:properties}.
\end{procedure}

In each one of the subphases of the second phase we iterate over all
facets of $P$ and treat each facet as a potential base facet of the
palm of a valid fixture (unless a fixture was found in a previous
subphase). In the following, we narrow down the search space for
fixtures with four fingers, once it has been established that our
workpiece does not have a fixture with two or three fingers.  Consider
a polyhedron $P$ that does have a valid fixture, say $G$ (with an
arbitrary number of fingers). There exists a subset
$\calR \subset H(\calF_{BT})$, such that
\begin{inparaenum}[(i)]
\item $\calR$ is a covering set of the closed hemisphere
  $\SOtwo\setminus H(\calF_{P})$, and
\item $|\calR|\in\{3,4,5\}$.
\end{inparaenum}
(This follows the same reasoning as in the proof of
Theorem~\ref{theorem:upper-bound}, which appears in
Appendix~\ref{sec:proofs}.)
The composition of $R$ can be categorized
into four cases listed below. We show that only one of theses cases,
namely Case IV, must be considered when searching for a fixture with
four fingers.

\begin{figure}[!ht]
  \vspace{-18pt}%
  \centering%
  \subfloat[]{\label{fig:tet}
    \includegraphics[width=2.0cm]{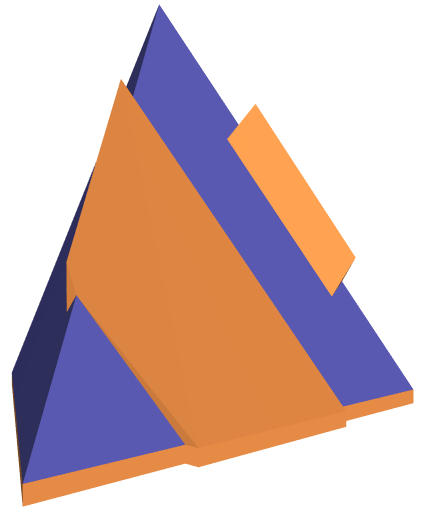}}\quad
  \subfloat[]{\label{fig:cube}
    \includegraphics[width=2.7cm]{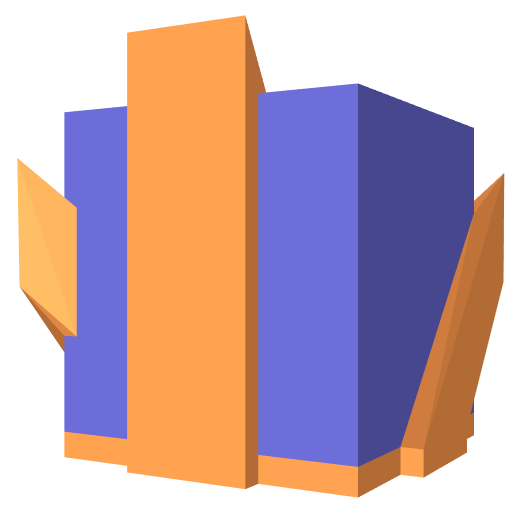}}\quad
  \subfloat[]{\label{fig:tri-prism}
    \includegraphics[width=2.0cm]{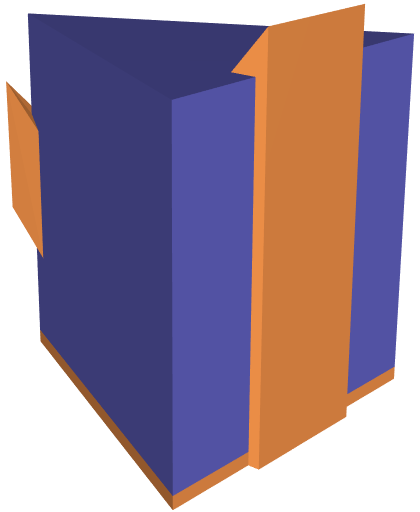}}\quad
  \subfloat[]{\label{fig:square-pyramid}
    \includegraphics[width=2.7cm]{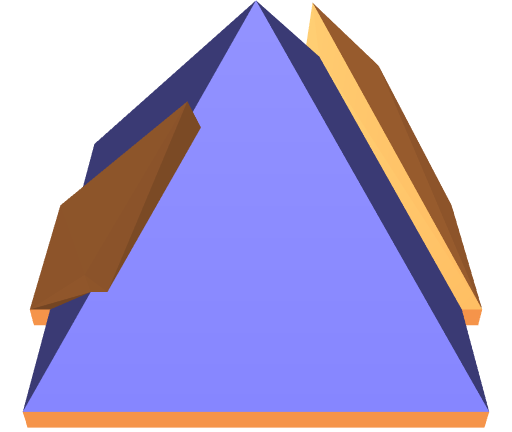}}
  \setlength{\abovecaptionskip}{3pt}
  \setlength{\belowcaptionskip}{3pt}
  \caption[]{%
  (\subref*{fig:tet}) A tetrahedron and a two-finger snapping fixture.
  (\subref*{fig:cube}) A cube and a three-finger snapping fixture.
  (\subref*{fig:tri-prism}) A triangular prism and a two-finger snapping fixture.
  (\subref*{fig:square-pyramid}) A square pyramid and a two-finger snapping
  fixture.}
  \label{fig:cases}%
  \vspace{-10pt}%
\end{figure}

Case I: $|\calR| = 3$. The tetrahedron and the fixture
depicted in Figure~\ref{fig:tet} demonstrate this case. At most three
distinct fingers of $G$ are needed; it implies that finding a fixture
similar to $G$, but only with these three fingers, during the first or
second subphases is guaranteed.

Case II: $|\calR| = 5$. The tetrahedron and the fixture
depicted in Figure~\ref{fig:cube} demonstrate this case.  By
Corollary~\ref{corollary:parallelepiped:3}, $\calR\cup H(\calF_{P})$
consists of three antipodal pairs of open unit hemispheres. As
$\calR\cup H(\calF_{P})$ is a covering set of $\SOtwo$ and $|\calR\cup
H(\calF_{P})|=6$, by Corollary~\ref{corollary:parallelepiped:3:no},
$H(\calF_{PBT}) =
\calR\cup H(\calF_{P})$. It implies that the facets in $\calF_{PBT}$
can be divided into three pairs of non-empty sets, such that each set
is a collection of all facets with the same normal, and the two sets
of every pair correspond to opposite normals, respectively. Without
loss of generality, we assume that $P$ does not have coplanar facets
that are neighbors, because such facets can be merged. Next, observe
that the facets in $\calF_{PBT}$ must be parallelograms. Assume, for
contradiction, that there exists a facet $f$ that is not a
parallelogram. It implies that $f$ has at least three neighboring
facets that are pairwise non-parallel, which implies that, together
with $h(f)$, $\calR$ contains at least four open hemispheres that are
pairwise non-antipodal, a contradiction.

$G$ must have at least one finger, say $F_1$, such that the normal to
the base facet of its fingertip, say $f_{t_1}$, is opposite to the
normal of the base facet of the palm $f_p$. Let $f_{b_1}$ denote the
base facet of the body of $F_1$. Consider the set $\calR_1 = \calR
\setminus \{h(f_{b_1}), h(f_{t_1})\}$. Observe that $|\calR_1|=3$.
Let $h(\bar{f_{b_1}})$ be the antipodal counterpart of $h(f_{b1})$.
Consider the finger $F_2$, such that $\bar{f_{b_1}}$ is
either the base facet, $f_{b_2}$, of the body of $F_2$ or the base
facet, $f_{t_2}$, of the fingertip of $F_2$. Naturally,
$h(\bar{f_{b_1}})$ is a member of $\calR_1$.
\begin{inparaenum}[(i)]
\item If $\bar{f_{b_1}} = f_{t_2}$, then, since $f_{b_2}$ is a
  neighbor of $f_p$ and $f_{t_2}$, $h(f_{b_2})$ must be a member of
  $\calR_1$ as well. Now, consider the set
  $\calR_2=\calR_1\setminus\{h(f_{b_2}),h(f_{t_2})\}$, and observe
  that $|\calR_2|=1$.
\item If $\bar{f_{b_1}} = f_{b_2}$, then let $f_{t'}$ be one of the
  neighbors of $f_{b_2}$ that is not parallel to $f_p$. Recall, that
  the facet $f_{b_2}$ has four neighbors---two pairs of parallel
  facets.  As $f_{b_2}$ and $f_{b_1}$ are parallel, $h(f_{t'})$ must
  be a member of $\calR_1$ as well. If $f_{t'} \neq f_{t_2}$, replace
  the fingertip of $F_2$ with a fingertip, the base of which is
  $f_{t'}$. Now, consider the set
  $\calR_2=\calR_1\setminus\{h(f_{b_2}),h(f_{t'})\}$, and observe that
  $|\calR_2|=1$.
\end{inparaenum}
It follows that there exists a third finger, say $F_3 \neq F_1,F_2$,
such that either $h(f_{b_3}) \in \calR_2$ or $h(f_{t_3})\in \calR_2$,
where $f_{b_3}$ and $f_{t_3}$ are the base facets of the body and
fingertip, respectively, of $F_3$, which obviates the need for further
fingers. It implies that finding a valid fixture during the first or
second subphases is guaranteed.

Case III: $|\calR| = 4$ and there exists a facet $f \in \calR$, such
that $h(f)$ and $h(f_p)$ are antipodal. The triangular prism and the
fixture depicted in Figure~\ref{fig:tri-prism} demonstrate
this case.  As in the previous case, $G$ must have at least one
finger, say $F_1$, such that the normal to the base facet of its
fingertip, say $f_{t_1}$, is opposite to the normal of the base facet of
the palm $f_p$. Let $f_{b_1}$ denote the base facet of the body of
$F_1$. Consider the set
$\calR_1=\calR\setminus\{h(f_{b_1}),h(f_{t_1})\}$.  Since
$|\calR_1|=2$, at most two additional distinct fingers of $G$ are
needed; it implies that finding a fixture similar to $G$, but only
with three fingers, during the first or second subphases is
guaranteed.

Case IV: $|\calR| = 4$ and $\calR$ does not contain an open
hemisphere, such that this hemisphere and $h(f_p)$ are antipodal. The
square pyramid and the fixture depicted in
Figure~\ref{fig:square-pyramid} demonstrate this case. Observe that
the fixture in the figure has two fingers. However, sometimes four
fingers are necessary as established by
Theorem~\ref{theorem:upper-bound}; see, e.g.,
Figure~\ref{fig:special:1}. This is the only case we need to consider
when searching for a fixture with four fingers. Notice, that in this
case, the intersections of at least two open hemispheres in $\calR$
with the great circle $\partial{h(f_p)}$ are pairwise antipodal open
unit semicircles.



\begin{procedure}\normalfont%
  (\textbf{\textsc{fourFingersFixtures}}$(f, \calC)$)
  The procedure accepts a facet $f$ of a potential palm and a set of
  pairs of facets, where each pair defines the base facets of the body
  and fingertip of a candidate finger, as input. It returns a valid
  fixture of $P$ with four fingers, if there exists one, such that $f$
  is the base facet of its palm, and its configuration matches Case
  IV above. Let $\calC'$ denote the set of unique
  facets in $\calC$. Let $\bar{h} = \SOtwo\setminus h(f)$ denote the
  closed hemisphere that must be covered by the open hemispheres
  $H(\calC')$. The procedure first divides all the hemispheres in
  $H(\calC')$ into equivalence classes, such that the intersections of
  all hemispheres in a class with the unit circle
  $C=\partial{\bar{h}}$ is a unique open semicircle. Let
  $s(\calE)=x\cap C, x\in \calE$ denote the unique open semicircle
  associated with the equivalence class $\calE$. There is a canonical
  total order of hemispheres within each class: Let $h_1$ and $h_2$ be
  two hemispheres in some class; then $h_1 \prec h_2$ iff
  $h_1\cap\bar{h}\subset h_2 \cap\bar{h}$.  Then, the procedure
  identifies pairs of equivalence classes $(\calE_1,\calE_2)$, such
  that $s(\calE_1)$ and $s(\calE_2)$ are antipodal open
  semicircles. For each pair, the procedure traverses all other
  equivalence classes twice searching for two additional equivalence
  classes $\calE_3$ and $\calE_4$, such that the set
  $\{s(\calE_1),s(\calE_2),s(\calE_3),s(\calE_4)\}$ covers $C$. If it
  finds such four equivalence classes, it implies that there exists a
  valid fixture with four fingers $F_1,F_2,F_3,F_4$, such that the
  maximal hemisphere associated with $\calE_i$ is either $h(f_{b_i})$
  or $h(f_{g_i})$. In this case the procedure returns such a fixture.
\end{procedure}

The complexity of the algorithm is the accumulated complexities of
Phase 1 and Subphases~2.1,~2.2, and~2.3. The efficiency (low
running-time complexity) of Subphase 2.2 stems from an observation on
the maximum number of possible candidates for this subphase, which in
turn relies on the genus of the polyhedron, as we discuss next.

\begin{lemma}[Genus of complete bipartite graphs~\cite{b-ongcpg-78}]
The genus of the complete bipartite graph, $k_{m,n}$, with $m$ nodes in one
side and $n$ in the other, is $\lceil(m-2)(n-2)/4\rceil$.\label{lemma:genus}
\end{lemma}

\begin{lemma}\label{lemma:num-palms}
Given an input polyhedron $P$ of genus $g$. Let $\tau$ be a triplet of
candidate fingers. Let $\calP$ be the set of plams, such that all
fingers in $\tau$ extend every palm in $\calP$. Then,
$|\calP| \leq 4\cdot g+2$.
\end{lemma}

\begin{proof}
Let $\calA$ be the set of three facets of $P$ that correspond to the
three base facets of the bodies of the fingers in $\tau$.  Let $\calB$
be the set of facets of $P$ that correspond to the base facets of the
palms in $\calP$. Let $V,E,F$ denote the vertices, edges, and facets
of $P$, respectively. Let $P^*=(V^*,E^*,F^*)$ be the dual graph of
$P$, where each facet is represented as a node, and two nodes are
connected by an arc if the corresponding two facets are neighbors.
According to Euler characteristic, the genus of $P^*$ is given by
$1-(|V^*|-|E^*|+|F^*|)/2$, which is equal to $1-(|F|-|E|+|V|)/2=g$.
Consider the subgraph $H$ of $P^*$ that consists of the nodes that
correspond to the facets in $\calA$ and in $\calB$. The genus of $H$
is at most $g$.  Since each facet in $\calA$ and each facet in $\calB$
are neighbors, $H$ is a complete bipartite graph $k_{(3,|{\calB}|)}$.
By Lemma~\ref{lemma:genus}, the genus of $H$ is
$\lceil(3-2)(|\calB|-2)/4\rceil=\lceil(|\calB|-2)/4\rceil \leq g$.
Hence, $|{\calB}| \leq g\cdot 4 + 2$.
\qed\end{proof}

\begin{theorem}\label{theorem:complexity}
  Algorithm~\ref{alg:snapfixt} runs in $O(n^{3})$ time, where $n$ is
  the number of vertices of the input polyhedron.
\end{theorem}

\begin{proof}
  During the first phase we list all the potential palms, each palm
  together with all the fingers that can be connected to it.  The
  overall number of potential fingers is twice the number of edges in
  the polytope; see
  Observation~\ref{observation:candidate-fingers}. Since The number of
  facets and the number of edges in a polytope with $n$ vertices is
  linear in $n$, the number of palm-finger combinations created in
  Phase~1 is $O(n^2)$.
  The second phase dominates the time complexity. We examine each
  subphase separately.  Recall, that a potential fixture passed to
  \textsc{validFixture}$(F)$ (encoded by $(f,S)$, where $f$ denotes a
  facet and $S$ denotes a set, the cardinality of which is fixed,
  i.e., 2, 3, or 4) has a fixed number of fingers. Therefore, every
  execution of the function consumes constant time.
  In the first subphase for every possible palm the function
  \textsc{validFixture} is invoked once per every subset of candidate
  fingers of size~2. As the number of candidate fingers is linear in
  $n$, the number of pairs of fingers is in $O(n^2)$. Thus, the total
  complexity of this subphase is $O(n\cdot n^2) = O(n^3)$.
  In the second subphase for every possible palm the function
  \textsc{validFixture} is invoked once per every subset of candidate
  fingers of size~3.  By Lemma~\ref{lemma:num-palms} and the
  assumption that the genus of the input polyhedron is constant, while
  iterating over all possible fixtures that have exactly three
  fingers, each triplet of fingers is considered a constant number of
  times. Therefore, the total time consumed processing potential
  fixtures of three fingers is bounded by $O(n^3)$.
  \textsc{fourFingersFixtures}($f, \calC$) is invoked once for every
  facet in the input polyhedron. Building the equivalence classes and
  finding the maximum of each class takes $O(n)$ time. Matching
  maximal hemispheres of equivalence classes to form pairs of
  associated antipodal semicircles takes $O(n^2)$ time.  Finally,
  examining every pair, traversing all other equivalence classes for
  each pair, also takes $O(n^2)$ time. Thus, the total complexity of
  this subphase is $O(n\cdot n^2) = O(n^3)$. The overall time
  complexity is thus $O(n^3)$.\qed
\end{proof}

\section{Two Applications}
\label{sec:cases}
We present two applications that utilize our algorithm and its implementation.
\subsection{Minimal Weight Fixtures}
\label{ssec:cases:weight}
\begin{figure}[!ht]
  \vspace{-4pt}
  \centering%
  \subfloat[]{\label{fig:ms:a}
    \includegraphics[height=2.5cm]{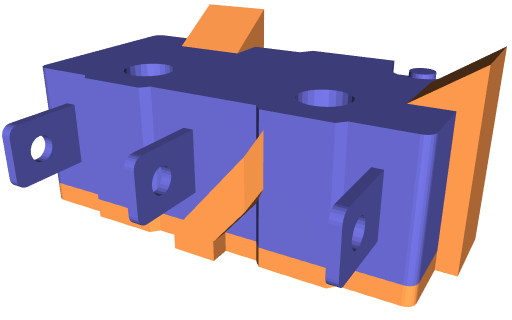}}\quad
  \subfloat[]{\label{fig:real-ms}
    \includegraphics[height=2.5cm]{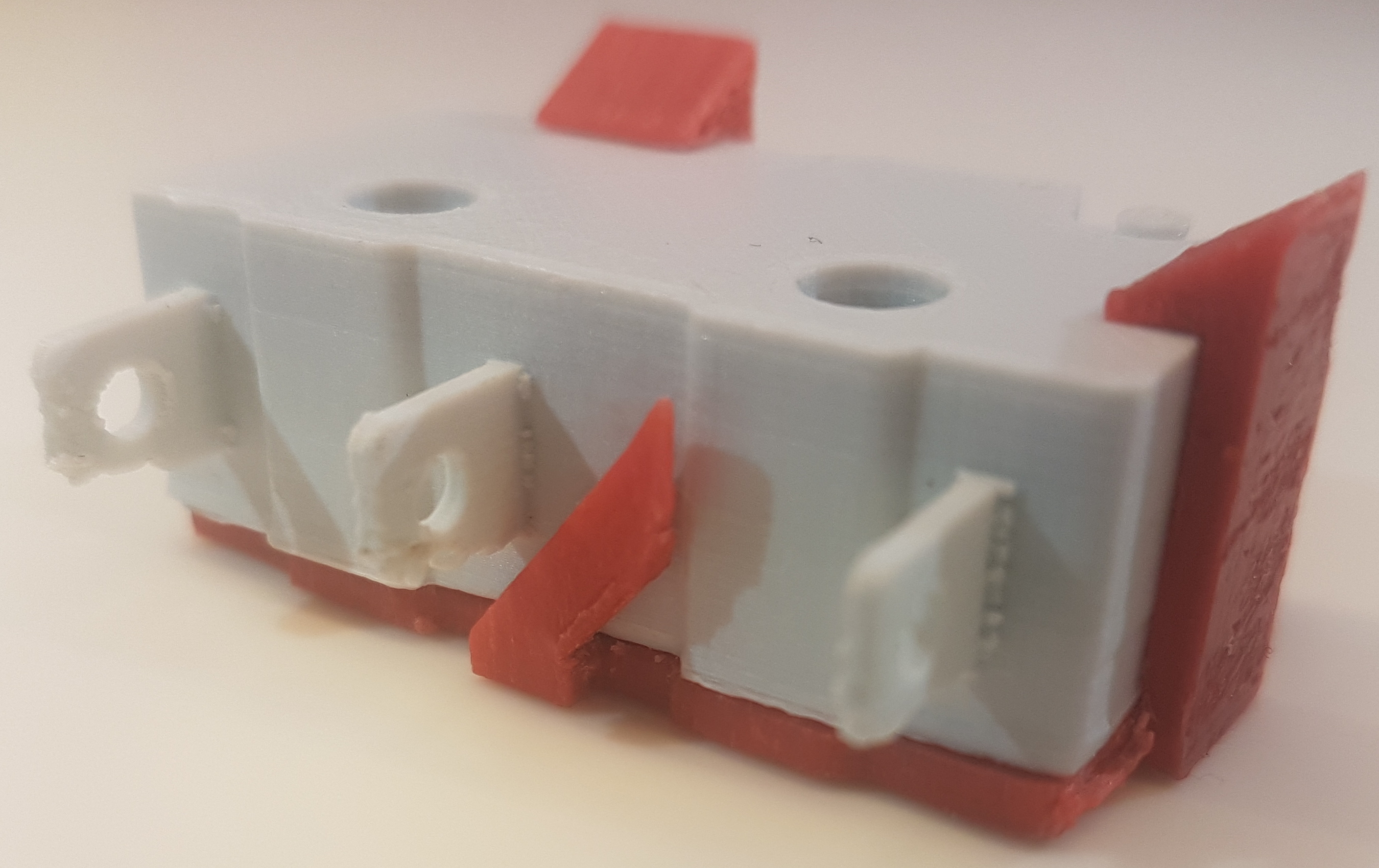}}\quad
  \subfloat[][]{\label{fig:real-drone}%
    \includegraphics[height=2.5cm]{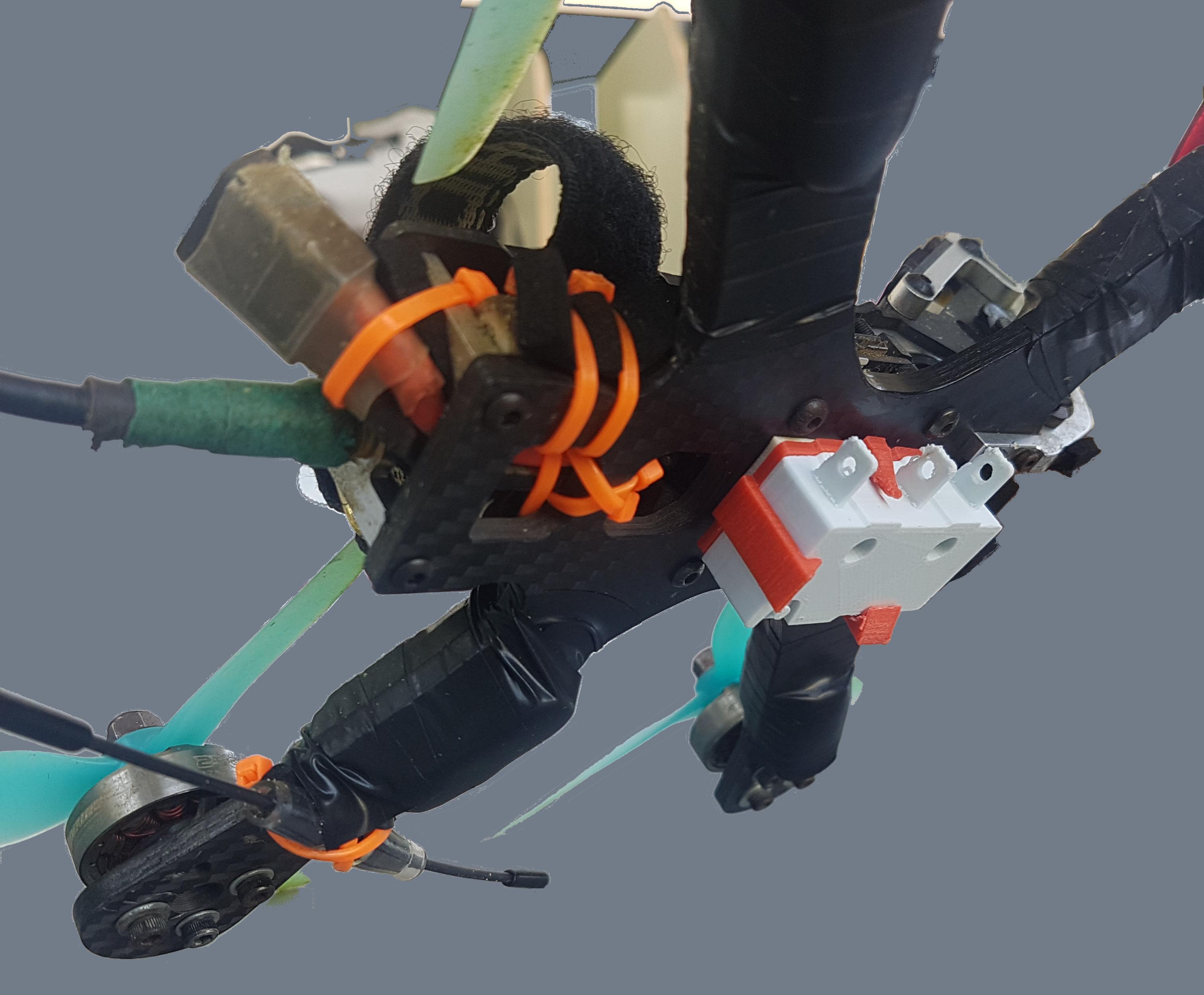}}
  \setlength{\abovecaptionskip}{0pt}
  \setlength{\belowcaptionskip}{-2pt}
  \caption[]{
    (\subref*{fig:ms:a}) Synthetic micro switch sensor and a snapping fixture assembled.
    (\subref*{fig:real-ms}) A micro-switch sensor held by a fabricated snapping fixture.
    (\subref*{fig:real-drone}) A drone with the snapping fixture attached to it.
  }
  \label{fig:ms}
  \vspace{-4pt}
\end{figure}
Generating lightweight fixtures that could be mounted on a UAV has
been a major challenge ever since the first UAV was introduced. The
desire for robust and efficient solutions to this problem rapidly
scaled up during the last decade with the introduction of small
drones, the weight of devices that can be mounted on which, is
limited.  Naturally, the device must be securely attached to the
drone; however, at the same time, the holding mechanism should weigh
as little as possible. Figure~\ref{fig:ms} shows a fixture generated
for a micro-switch sensor, a common sensor in the field
of robotics and automation. Figure~\ref{fig:real-drone} shows the
fabricated fixture (3D printed) permanently attached to a drone. It
holds a micro-switch. While the micros-switch is firmly held during
flight, it can be easily replaced.

\subsection{Minimal Obscuring Fixtures}
\label{ssec:cases:visibility}
One of the objectives of jewelry making is to expose the gems mounted
on a jewel, such as a ring, and reveal their allure. As with the
minimal-weight fixture, the mounted gem must be securely attached to
the jewel; however, the weight of the holding mechanism can be
compromised. Here we seek to find a fixture that obscures the gem as
little as possible, so that the gem surface is exposed as much as
possible. Figure~\ref{fig:spider} shows a pendant with an integrated
fixture synthesized by our generator. The fixture in
Figure~\ref{fig:emerald} is generated for an emerald cut; it reveals a
surprising portion of the front facet of the stone.%

\begin{figure}[!ht]
  \vspace{-4pt}
  \centering%
  \subfloat[][]{\label{fig:emerald}%
    \includegraphics[height=2.5cm]{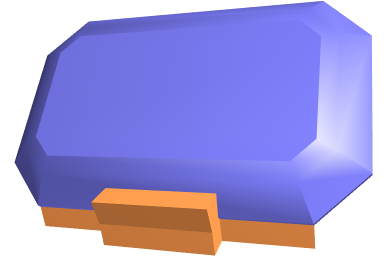}}\quad
  \subfloat[][]{\label{fig:spider}%
    \includegraphics[height=2.5cm]{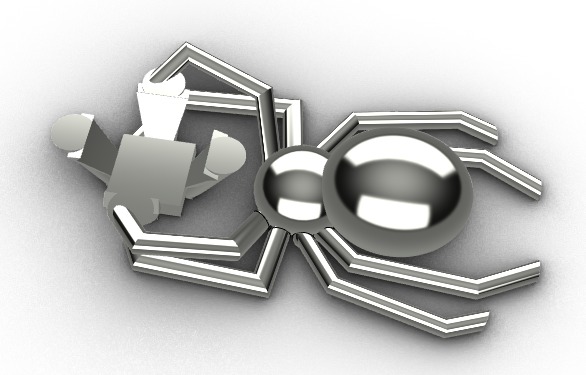}}\quad
  \subfloat[][]{\label{fig:real-spider}%
    \includegraphics[height=2.5cm]{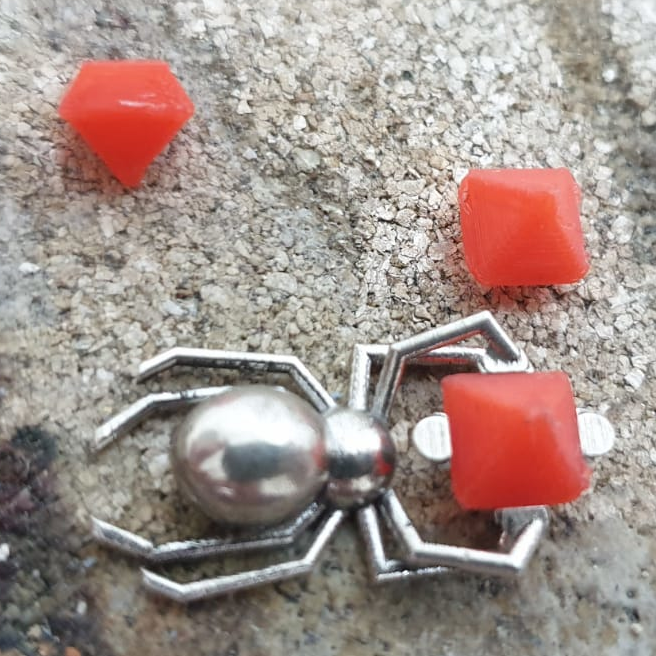}}
  \setlength{\abovecaptionskip}{0pt}
  \setlength{\belowcaptionskip}{-2pt}
  \caption[]{(\subref*{fig:emerald}) an emerald cut---a common cut for precious
    stones.
    (\subref*{fig:spider}) a synthetic pendant with an integrated snapping
    fixture.
    (\subref*{fig:real-spider}) the fabricated pedant holding a precious stone.}
  \label{fig:jewel}
  \vspace{-4pt}
\end{figure}


\subsection{3D Printing Considerations}
\label{ssec:cases:considerations}
We used various materials for generating snapping fixtures, such as,
ABS, PLA, PETG, Nylon 12, and Sterling silver.\footnote{3D printed wax
  and lost-wax where used to generate fixtures made of Sterling
  Silver.} All generated fixtures properly snapped and firmly held the
workpieces. However, low quality prints (made of ABS, PLA, or PETG)
occasionally broke after repeated or incautious uses. We noticed that
increasing the infill density and orienting the prints such that the
joint axes and the printing plate are not parallel increase the fixture
durability. Also, we compensated for the limited precision of printers
by scaling up the fixture to create a gap of up to 0.2mm between the
fixture and the workpiece.

\section{Experimental Results}
\label{sec:experiments}
The generator was developed in C++; it depends on the Polygon Mesh
Processing package of \cgal~\cite{cgal:lty-pmp-20a}. Table~\ref{tab:experiments} lists some of the workpieces
we fed as input, and provides information about the generation of the
corresponding snapping fixtures. The coordinates of the vertices of
the input models were given in floating point numbers. The generator was
executed on an \emph{Intel Core i7-2720QM} CPU clocked at~2.2~GHz
with~16~GB of RAM.

\begin{table}[!ht]
  \scriptsize
  \renewcommand{\arraystretch}{1.2}
  \newlength{\cellWidth}\setlength\cellWidth{9.5ex}
  \label{tab:experiments}
  \caption{Information related to snapping fixture generation of various
    workpieces. \textbf{Verts}, \textbf{Tris}, and \textbf{Fixts} stand for
    \textbf{Vertices}, \textbf{Triangles}, and \textbf{Fixtures}, respectively.
    The column entitled \textbf{Merged} indicates the number of facets
    after the merging of coplanar triangular facets.
    The last column indicates the number of fixtures that admit the minimal
    number of fingers.}
  \centering
  \begin{tabular}{|l|>{\raggedleft}p{\cellWidth}|>{\raggedleft}p{\cellWidth}|>{\raggedleft}p{\cellWidth}|>{\raggedleft}p{\cellWidth}|>{\raggedleft}p{\cellWidth}|>{\raggedleft}p{\cellWidth}|>{\raggedleft}p{\cellWidth}|r|}
    \hline
    \multicolumn{6}{|c|}{\textbf{\mbox{Workpiece}}} &
    \multicolumn{2}{c|}{\textbf{\mbox{Fixture}}} &
    \multicolumn{1}{c|}{\textbf{\mbox{\# Fixts}}}\\\cline{1-6}\cline{7-8}
    \multirow{2}{*}{\textbf{\mbox{Name}}} &
    \multicolumn{1}{c|}{\textbf{\mbox{\#}}} &
    \multicolumn{1}{c|}{\textbf{\mbox{\#}}} &
    \multicolumn{2}{c|}{\textbf{\mbox{\# Facets}}} &
    \multirow{2}{*}{\textbf{\mbox{Genus}}} &
    \multicolumn{1}{c|}{\textbf{\mbox{\# Min}}} &
    \multicolumn{1}{c|}{\textbf{\mbox{Time}}} &
    \multicolumn{1}{c|}{\textbf{\mbox{Min}}}\\\cline{4-5}
    &
    \multicolumn{1}{c|}{\textbf{\mbox{Verts}}} &
    \multicolumn{1}{c|}{\textbf{\mbox{Edges}}} &
    \multicolumn{1}{c|}{\textbf{\mbox{Tris}}} &
    \multicolumn{1}{c|}{\textbf{\mbox{Merged}}} &
    &
    \multicolumn{1}{c|}{\textbf{\mbox{Fingers}}} &
    \multicolumn{1}{c|}{\textbf{\mbox{(ms)}}} &
    \multicolumn{1}{c|}{\textbf{\mbox{Fingers}}}\\
    \hline
    \hline
    tetrahedron         &   4 &     6 &     4 &   4 & 0 & 2 &      3 &      36\\
    dodecahedron\footnotemark & 20 & 30 &  36 &  12 & 0 & 2 &     15 &      50\\
    emerald             &  34 &    96 &    64 &  25 & 0 & 2 &     39 &       8\\
    square pyramid      &   5 &     8 &     6 &   5 & 0 & 2 &      4 &      24\\
    micro switch        & 594 & 1,806 & 1,204 & 305 & 2 & 2 & 42,761 & 263,895\\
    cube                &   8 &    18 &    12 &   6 & 0 & 3 &     20 &     216\\
    octahedron          &   6 &    12 &     8 &   8 & 0 & 3 &      3 &      16\\
    torus               &  32 &    64 &    32 &  10 & 1 & 3 &    307 &   2,760\\
    4-finger            &  26 &    64 &    42 &  41 & 0 & 4 &     45 &      17\\
    truncated cuboctahedron\footnotemark[\value{footnote}] & 48 & 72 & 92 & 26 & 0 & 2 & 163 & 29\\
    icosahedron         &  12 &    30 &    20 &  20 & 0 & $\infty$ & 22 &    0\\
    8-base cylinder     &  16 &    42 &    28 &  10 & 0 & 2 &     44 &     106\\
    28-base cylinder    &  56 &   162 &   108 &  30 & 0 & 2 &    984 &   4,396\\
    48-base cylinder    &  96 &   282 &   188 &  50 & 0 & 2 &  4,672 &  24,456\\
    68-base cylinder    & 136 &   402 &   268 &  70 & 0 & 2 & 13,008 &  71,892\\
    88-base cylinder    & 176 &   522 &   348 &  90 & 0 & 2 & 27,233 & 159,124\\
    108-base cylinder   & 216 &   642 &   428 & 110 & 0 & 2 & 50,122 & 297,956\\
    \hline
  \end{tabular}
  \vspace{-6pt}
\end{table}
\footnotetext{Limited precision coordinates render the actual models non-regular.}

\vspace{\beforeSectionVSpace}
\section*{Acknowledgement}
\label{sec:ack}
\vspace{\afterSectionVSpace}
The authors thank Shahar Shamai for his technical support, Omer Kafri and Raz
Parnafes for fruitful discussions, and Shahar Deskalo for exploiting our work
while crafting a jewel.

\pagebreak
\bibliographystyle{unsrt}
\bibliography{../abrev-short,../snapping_fixture_generation}
\appendix
\section{Notation Glossary}
\label{sec:glossary}
The following lists typical notations.
\begin{itemize}
\item $\openCS$, $\closedCS$---general open and closed unit semicircles, respectively
\item $\openHS$, $\closedHS$---general open and closed unit hemispheres, respectively
\item $\PPone$---the affinely extended real number line
\item $\PPtwo$---a generalization of the affinely extended real number line to the plane
\item $\SOone$---the unit circle
\item $\SOtwo$---the unit sphere
\item $f$---a facet
  \begin{itemize}
  \item $f_p$---the base facet of the palm of a fixture
  \item $f_{b_i}$---the base facet of the body of finger $i$
  \item $f_{t_i}$---the base facet of the fingertip of finger $i$
  \end{itemize}
\item $s$, $\bar{s}$---instances of an open and a closed unit semicircle, respectively
\item $h$, $\bar{h}$---instances of an open and a closed unit hemispheres, respectively
\item $\calC$, $\calE$, $\calR$, $\calS$---sets
\item $\calF$---a set of facets
  \begin{itemize}
  \item $\calF_P$---the singleton that consists of the base facet of the palm of
    a fixture
  \item $\calF_B$---the set of the base facets of the bodies of the fingers
    of a fixture
  \item $\calF_T$---the set of the base facets of the fingertips of the fingers
    of a fixture
  \item $\calF_{PBT}$---the union of the above three sets
  \item $\calF^P$---the facets of a polyhedron $P$
  \end{itemize}
\item $P$,$G$---polyhedrons, a workpiece and a snapping fixture, respectively
\item $h(f)$---a mapping from a facet to the hemisphere that consists of the
  blocking directions induced by $f$
\item $H(\calF)$---a mapping from a set of facets to the corresponding
  hemispheres
\end{itemize}

\section{Simple Algorithm}
\label{sec:simple-algorithm}
\begin{procedure}{\sc snappingFixture($P$)}
The procedure accepts a polyhedron $P$ as input and returns a fixture
$G$ of $P$ of the best quality according to given optimization
criteria (see below); see Algorithm~\ref{alg:snapfixt}. The algorithm
consists of two phases.  In the first phase we compute a data
structure $M$ that associates palms and their candidate fingers. In
the second phase we identify subsets of fingers for each palm stored
in $M$ that together form a potential valid fixture and extract the
fixture with the best quality over all potential fixtures.
\end{procedure}

\begin{algorithm}[htb]
  \caption{Snapping fixture generation}
  \label{alg:simple-snapfixt}
  \begin{algorithmic}[1]
    \Require A simple polyhedron $P$ that consists of $m$ facets
      $\{f_1,f_2,...,f_m\}$.
    \Ensure A snapping fixture $G$ of $P$, if there exists one, with the
    best quality.
    \Procedure{snappingFixture}{$P$}
      \For{$i\gets 1, m$}
        \State{$M[i]\gets \emptyset$}
        \ForAll{$j, f_j \in \Call{neighbours}{f_i}$}
          \ForAll{$\ell, f_{\ell} \in \Call{neighbours}{f_j}\,\&\,\ell \neq i$}
            \State{$M[i]\gets M[i] \cup \{(j,\ell)\}$}
          \EndFor
        \EndFor
      \EndFor
      \State{$B\gets \textrm{null}$}\Comment{Initialize best fixture}
      \For{$i\gets 1, m$}
        \ForAll{$S, S \in \Call{subsets}{M[i], 4}$}\Comment{$|S| \leq 4$}
          \State{$F\gets (f_i,S)$}\Comment{Define a fixture}
          \If{$\Call{validFixture}{F}$}
            \State{$B\gets\Call{findBest}{B, F}$}
          \EndIf
        \EndFor
      \EndFor\\
      \Return{$F$}
    \EndProcedure
  \end{algorithmic}
\end{algorithm}

The following is a generalization of Lemma~\ref{lemma:num-palms}.

\begin{lemma}\label{lemma:num-palms-general}
Given an input polyhedron $P$ of genus $g$. Let $\calQ$ be a set of at
least three candidate fingers. Let $\calP$ be the set of plams, such
that all fingers in $\calQ$ extend every palm in $\calP$. Then,
$|\calP| \leq \frac{4\cdot g}{|\calQ|-2}+2$.
\end{lemma}

\begin{proof}
  Let $\calB$ be the set of facets of $P$ that correspond to the base
  facets of the palms in $\calP$.  Following the same reasoning of the
  proof of Lemma~\ref{lemma:num-palms},
  $\lceil(|\calQ|-2)(|\calB|-2)/4\rceil\leq g$.  Hence, $|{\calB}|
  \leq \frac{4\cdot g}{|\calQ|-2}+2$.  \qed\end{proof}

\begin{theorem}\label{theorem:complexity}
Algorithm~\ref{alg:snapfixt} runs in $O(n^{4})$ time, where $n$ is the
number of vertices of the input polyhedron.
\end{theorem}

\begin{proof}
  The first phase of Algorithm~\ref{alg:simple-snapfixt} is identical
  to the first phase of Algorithm~\ref{alg:snapfixt}; thus, the
  overall time complexity of the first phase is $O(n^2)$.  The second
  phase dominates the time complexity. Recall that a potential fixture
  passed to {\sc validFixture}$(F)$ and to {\sc findBest}$(B,F)$ has
  at most~4 fingers. Therefore, both functions run in constant
  time. For every possible palm the function {\sc validFixture} is
  invoked once per every subset of candidate fingers of size at
  most~4. Recall that the number of candidate fingers is linear in
  $n$; see Observation~\ref{observation:candidate-fingers}. Therefore,
  the number of subsets of size two is $O(n^2)$, and the total time
  consumed processing potential fixtures of two fingers is bounded
  by $O(n^3)$. We restate the assumption that the genus of the input
  polyhedron is constant, as the following deductions depend on it. By
  Lemma~\ref{lemma:num-palms}, each triplet of fingers is considered a
  constant number of times. Therefore, the total time consumed
  processing potential fixtures of three fingers is bounded by
  $O(n^3)$. Similarly, By Lemma~\ref{lemma:num-palms-general} each
  quadruplet of fingers is considered a constant number of
  times. Therefore, the total time consumed processing potential
  fixtures of three fingers is bounded by $O(n^4)$. Thus, the total
  complexity of this phase is $O(n^{4})$.

\end{proof}

\section{Proofs}
\label{sec:proofs}

\begin{theorem}[Helly's theorem~\cite{hv-httgt-17}]
Let $\calS=\{ X_{1},...,X_{n}\} $ be a finite collection of convex
subsets of $\Rd$, with $n > d$.  If the intersection of every $d + 1$
of these sets is nonempty, then the whole collection has a nonempty
intersection; that is, $\cap_{j=1}^{n}X_{j} \neq \emptyset$.
\end{theorem}

The contrapositive formulation of the theorem follows.  If
$\cap_{j=1}^{n}X_{j}=\emptyset$ then there exists a subset $\calR=\{
X_{i_{1}},...,X_{i_{d+1}}\}\subseteq \calS$ such that $|\calR|=d+1$
and $\cap_{j=1}^{d+1}X_{i_{j}}=\emptyset$.  In the succeeding proofs
we use the following corollary:

\begin{corollary}
Let $\calS=\{ X_{1},...,X_{n}\} $ be a finite set of convex subsets of
$\Rd$. If $\cup_{j=1}^{n}X_{j}=\Rd$ then there exists a subset
$\calR=\{ X_{i_{1}},...,X_{i_{d+1}}\}\subseteq \calS$ such that
$|\calR|=d+1$ and $\cup_{j=1}^{d+1}X_{i_{j}}=\Rd$.
\end{corollary}

The corollary holds because the intersection of a set of subgroups of
$\Rd$ is empty iff the union of their complement in $\Rd$ is $\Rd$.

The following four lemmas, namely, \ref{lemma:A}--\ref{lemma:D},
are based on the analysis in~\cite{shb-spspm-17}.

\begin{lemma}\label{lemma:A}
Let $\calS$ be a finite set of open unit semicircles. If $\calS$ is a covering
set of a closed unit semicircle $\closedCS$, then there exists
$\calR\subseteq \calS$ such that $\calR$ is a covering set of $\closedCS$ and
$|\calR|\in\{2,3\}$.
\end{lemma}

\begin{proof}
  It is obvious that one open unit semicircle cannot cover a closed
  unit semicircle. Let $\openCS$ denote the interior of $\closedCS$.
  ($\openCS$ is an open unit semicircle.) There are two cases:
  (i) $\openCS \in \calS$ and $\openCS \in \calR$ for every
  covering set $\calR\subseteq \calS$ of $\closedCS$. It implies that every
  covering set $\calR\subseteq \calS$ must contain two additional open
  semicircles that cover the two boundary points of $\closedCS$,
  respectively. These two semicircles together with $\openCS$
  constitute a covering set of $\closedCS$ of size three.
  (ii) There exists a covering set $\calS' \subseteq \calS$, where
  $\openCS \notin \calS'$. Let $\calS'_{\closedCS}=\{s\cap \closedCS\,|\,s\in
  \calS'\}$ be the set of intersections of the elements of $\calS'$ and
  $\closedCS$.
  Let $\Pi^1$ denote the extended central projection that maps the
  closed semicircle $\closedCS$ to the \emph{affinely extended real
    number line},\footnote{The set $\Rone \cup \{+\infty,-\infty\}$ is
    referred to as the affinely extended real number line.}, $\Pi^1(p) =
  (x,w):\closedCS \rightarrow \PPone$, where the points in \PPone are
  represented in homogeneous coordinates $(x,w)$.
  Notice that for every $s\in \calS'_{\closedCS}$, $s$ covers one of the
  boundary points of $\closedCS$; therefore, $\Pi^1(s)$ is an open ray
  covering either $(-1,0)$ or $(+1,0)$. $\calS'_{\closedCS}$ covers
  $\closedCS$; therefore, the set of its images
  $\calS'_{\Pi^1}=\{\Pi^1(s)\,|\,s\in \calS'_{\closedCS}\}$ covers \PPone. By
  Helly's theorem, there exists a subset $\calR'_{\Pi^1} \subseteq \calS'_{\Pi^1}$
  of size two that covers \Rone. Thus, the set of preimages of
  $\calR'_{\Pi^1}$ covers $\closedCS$.
\qed\end{proof}

\begin{lemma}\label{lemma:B}
  Let $\calS$ be a finite set of open unit semicircles. If $\calS$ is
  a covering set of the unit circle $\SOone$, then there exists $\calR
  \subseteq \calS$ such that $\calR$ is a covering set of $\SOone$ and
  $|\calR| \in \{3,4\}$.
\end{lemma}

\begin{proof}
  Let $s \in \calS$ be an arbitrary open unit semicircle in $\calS$. The
  remaining elements $\calS \setminus \{s\}$ of $\calS$ must cover the
  complement $\bar{s}$ of $s$, which is a closed unit semicircle.
  By lemma~\ref{lemma:A}, there exists $\calR' \subseteq \calS \setminus \{s\}$
  that covers $\bar{s}$, and $|\calR'| \in \{2,3\}$. Thus, $\calR' \cup
  \{s\}$ covers $\SOone$, and $|\calR' \cup \{s\}| \in \{3,4\}$.
\qed\end{proof}


\begin{lemma}\label{lemma:C}
Let $\calS$ be a finite set of open unit hemispheres. If $\calS$ is a covering
set of a closed unit hemisphere $\closedHS$, then there exists
$\calR \subseteq \calS$, such that $\calR$ is a covering set of $\closedHS$ and
$|\calR|\in\{3,4,5\}$.
\end{lemma}

\setlength{\intextsep}{0pt}%
\begin{wrapfigure}[8]{R}{3.8cm}
  \begin{tikzpicture}[node distance=5cm]
    \node (img) {\includegraphics[width=3.8cm]{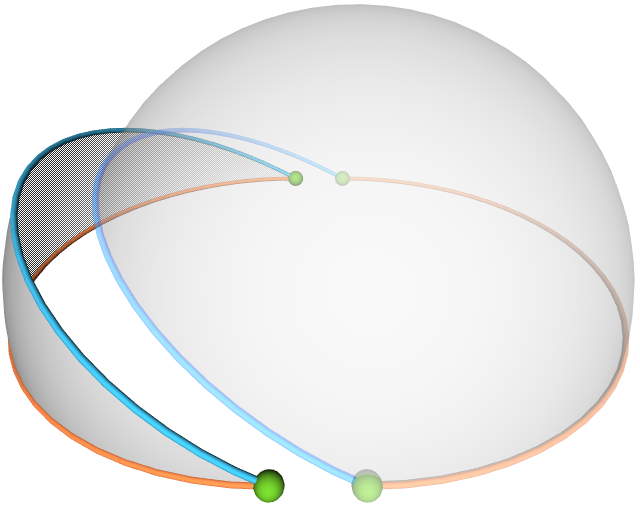}};
    \node [](B) at (-1.6, -1.2) {\textcolor{color1!80!black}{$\closedCS_1$}};
    \node [](A) at (-1.6, 0.95) {\textcolor{color2!80!black}{$\closedCS_2$}};
  \end{tikzpicture}
\end{wrapfigure}
\noindent\textit{Proof.}
  Let $\openHS$ denote the interior of $\closedHS$ ($\openHS$ is
  an open unit hemisphere) and $\boundaryHS$ denote the boundary of
  $\closedHS$ ($\boundaryHS$ is a great circle).
  Similar to the proof of Lemma~\ref{lemma:A}, there are two cases:
  (i) $\openHS \in \calS$ and $\openHS \in \calR$ for every covering
  set $\calR \subseteq \calS$ of $\closedHS$. It implies that every
  covering set $\calR \subseteq \calS$ must contain additional open
  hemispheres that cover $\boundaryHS$.
  Let $\calS_{\boundaryHS}=\{s\cap \boundaryHS\,|\,s\in \calS\}$ be
  the set of intersections of the elements of $\calS$ and
  $\boundaryHS$. Note that an intersection of a unit open hemisphere
  and a great circle is either empty or an open unit
  semicircle. Therefore $\calS_{\boundaryHS}$ is a set of open unit
  semicircles lying on the same plane. By Lemma~\ref{lemma:B}, there
  exists a covering set $\calR_{\boundaryHS} \subset
  \calS_{\boundaryHS}$ of $\boundaryHS = \SOone$, such that
  $|\calR_{\boundaryHS}\,|\,\in \{3,4\}$.  This implies that there
  exists a covering set $\calR \subseteq \calS$ of $\closedHS$, such
  that $|\calR| \in \{4,5\}$.
  (ii) There exists a covering set $\calS' \subseteq \calS$, where
  $\openHS \notin \calS'$.  Let $\calS'_{\closedHS}=\{s\cap
  \closedHS\,|\,s\in \calS'\}$ be the set of intersections of the
  elements of $\calS'$ and $\closedHS$.
  Let $\Pi^2$ denote the extended central projection that maps the
  closed hemisphere $\closedHS$ to an extended plane obtained by
  adjoining all signed slopes to \Rtwo (a generalization of the
  affinely extended real number line, to the plane), $\Pi^2(p) =
  (x,y,w):\closedHS \rightarrow \PPtwo$, where the points in \PPtwo
  are represented in homogeneous coordinates $(x,y,w)$.
  Notice that every $s\in \calS'_{\closedHS}$ is a semi-open spherical
  wedge; see the figure in the previous page. The wedge is bounded
  by two semicircles $\closedCS_1$ and $\closedCS_2$ (in the figure),
  where $\closedCS_1$ lies in $\boundaryHS$. The intersection of
  $\closedCS_2$ and $s$ is empty, and the intersection of
  $\closedCS_1$ and $s$ is an open semicircle; therefore, $\Pi^2(s)$ is
  an open halfplane. $\calS'_{\closedHS}$ covers $\closedHS$; therefore,
  the set of its images $\calS'_{\Pi^2}=\{\Pi^2(s)\,|\,s\in \calS'_{\closedHS}\}$
  covers \PPtwo. By Helly's theorem, there exists a minimal subset
  $\calR'_{\Pi^2} \subseteq \calS'_{\Pi^2}$ of size at most three that covers
  \Rtwo.  If $|\calR'_{\Pi^2}|=2$, that is, two open halfplanes, say $h_1$
  and $h_2$ comprise $\calR'_{\Pi^2}$, then they must be parallel: $h_1: a x
  + b y + c_1 > 0$ and $h_2: a x + b y + c_2 > 0$.  In this case they
  do not cover the points $(-b, a)$ and $(b, -a)$ in $\PPtwo$.  Thus,
  the pair of preimages of $\calR'_{\Pi^2}$ covers $\closedHS$ except for
  two antipodal points. Covering these antipodal points requires two
  additional elements from $\calS'_{\closedHS}$, which yields a covering
  set of size four. If $|\calR'_{\Pi^2}| = 3$, then none of the halfplanes in
  $\calR'_{\Pi^2}$ (which cover $\Rtwo$) are parallel, and they also cover
  $\PPtwo$.  Thus, the set of preimages of $ \calR'_{\Pi^2}$ covers
  $\closedHS$, which yields a covering set of size three.\qed
\setlength\intextsep{\intextsepSaved}%
\setlength\columnsep{\columnsepSaved}%

\begin{lemma}\label{lemma:D}
  Let $\calS$ be a finite set of open unit hemispheres. If $\calS$ is a
  covering set of the unit sphere $\SOtwo$, then there exists
  $\calR\subseteq \calS$ such that $\calR$ is a covering set of $\SOtwo$ and
  $|\calR|\in\{4,5,6\}$.
\end{lemma}

\begin{proof}
  Let $s \in \calS$ be an arbitrary open unit hemisphere in $\calS$. The
  remaining elements $\calS \setminus \{s\}$ of $\calS$ must cover the
  complement $\bar{s}$ of $s$, which is a closed unit hemisphere.
  By lemma~\ref{lemma:C}, there exists $\calR' \subseteq \calS \setminus \{s\}$
  that covers $\bar{s}$, and $|\calR'| \in \{3,4,5\}$. Thus,
  $\calR' \cup \{s\}$ covers $\SOtwo$, and $|\calR' \cup \{s\}| \in \{4,5,6\}$.
\qed\end{proof}

\noindent%
Proof of Corollary~\ref{corollary:parallelepiped:2:minimal}.
\begin{proof}
  $(\Rightarrow)$ Assume, by contradiction, that $\calR$ contains an open
  unit semicircle $a$, such that the interior of its complement is not
  in $\calR$. Observe that the complement of $a$ is a closed unit
  semicircle. This is exactly case (ii) in the proof of
  Lemma~\ref{lemma:A}. Here, there exists a covering set $\calR'$ of the
  closed unit semicircle, such that $|\calR'|=2$. It implies that $|\calR|$
  is at most three, a contradiction.

  \medskip\noindent $(\Leftarrow)$ If $\calR$ consist of two antipodal
  pairs of open unit hemispheres, then the removal of any one of the
  four hemispheres leaves one point on $\SOone$ uncovered.
  \qed\end{proof}

\noindent%
Proof of Corollary~\ref{corollary:parallelepiped:2:no}.
\begin{proof}
  Assume, for contradiction, that a subset $\calR \subset \calS$,
  $|\calR|=3$ that covers $\SOone$ does not exist. By
  Lemma~\ref{lemma:B}, there exists a minimal subset $\calR$ of
  $\calS$ that covers $\SOone$ and $|\calR|=4$.  By
  Corollary~\ref{corollary:parallelepiped:2:minimal}, $\calR$ consists
  of two antipodal pairs of open unit semicircles.  Let $\bar{a}$
  denote the complement of the sole semicircle in $\calS \setminus
  \calR$. Observe that $\bar{a}$ is equivalent to the closed
  semicircle $\closedCS$, and that $a$, the interior of $\bar{a}$, is
  not in $\calR$.  This, again, is exactly case (ii) in the proof of
  Lemma~\ref{lemma:A}.  Here, there exists a covering set $\calR'$ of
  $\closedCS$, such that $|\calR'|=2$.  It implies that $|\calR|=3$, a
  contradiction.
  \qed\end{proof}

\noindent%
Proof of Corollary~\ref{corollary:parallelepiped:3}.
\begin{proof}
  $(\Rightarrow)$ Assume, by contradiction, that $\calR$ contains an open
  unit hemisphere $a$, such that the interior of its complement is not
  in $\calR$. Observe that the complement of $a$ is equivalent to
  $\HOtwo$. This is exactly case (ii) in the proof of
  Lemma~\ref{lemma:C}. Here, there exists a covering set $\calR'$ of
  $\HOtwo$, such that $|\calR'| \in \{3,4\}$. It implies that $|\calR|$ is at
  most five, a contradiction.

  \medskip\noindent$(\Leftarrow)$ $\calR$ consists of three antipodal
  pairs of open unit hemispheres that cover $\SOtwo$. Arbitrarily pick
  one antipodal pair. There is a great circle $c$ that it not covered
  by the pair. By corollary~\ref{corollary:parallelepiped:2:minimal}
  two antipodal pairs of open unit semicircles are required to cover
  $c$; they must be the intersections of the remaining two antipodal
  pairs of open unit hemispheres, respectively. Thus, six open
  hemispheres are required in total.
  \qed\end{proof}

\noindent%
Proof of Corollary~\ref{corollary:parallelepiped:3:no}.
\begin{proof}
  Assume, for contradiction, that a subset $\calR \subset \calS$,
  $|\calR|=5$ that covers $\SOtwo$ does not exist. By
  Lemma~\ref{lemma:B}, there exists a minimal subset $\calR$ of
  $\calS$ that covers $\SOtwo$ and $|\calR|=6$. By
  Corollary~\ref{corollary:parallelepiped:3}, $\calS$ consists of
  three antipodal pairs of open unit hemispheres. Let $\bar{h}$ denote
  the complement of the sole hemisphere in $\calS \setminus
  \calR$. Observe that $\bar{h}$ is equivalent to $\closedHS$, and
  that $h$, the interior of $\bar{h}$, is not in $\calR$. This, again,
  is exactly case (ii) in the proof of Lemma~\ref{lemma:C}.  Here,
  there exists a covering set $\calR'$ of $\closedHS$, such that
  $|\calR'|\in\{3,4\}$. It implies that $|\calR|\leq 5$, a
  contradiction.
  \qed\end{proof}

\medskip
\noindent%
Proof of Observation~\ref{observation:candidate-fingers}.
\begin{proof}
  Let $e$ be an edge of $P$ and let $f_{e}$ and $f_{e}'$ be the two
  faces incident to $e$. Two fingers can be built on $e$. The base
  facet of the body and the base facet of the tip of one finger
  coincides with $f_{e}$ and $f_{e}'$, respectively. In order to
  construct the other finger, the roles of these facets exchange; that
  is, the base facet of the body and the base facet of the fingertip
  coincides with $f_{e}'$ and $f_{e}$, respectively. Every candidate
  finger is built on a single edge. Thus, the number of candidate
  fingers is at most $2|E|$. From Euler's formula we know that the
  number of edges in a polyhedron is linear in the number of vertices
  of the polyhedron. Thus, the number of candidate fingers is at most
  $6n-12$.
\qed\end{proof}

\medskip
\noindent%
Proof of Theorem~\ref{theorem:upper-bound}
\begin{proof}
  Consider a polyhedron $P$. Let $G$ be a valid fixture of $P$, and
  assume that $G$ has more than four fingers. We show that it is
  possible to construct a valid snapping fixture of $P$ that has (i)
  the same palm as $G$, and (ii) four fingers that are a subset of the
  fingers of $G$.  Consider the closed hemisphere $\closedHS =
  \SOtwo\setminus H(\calF_{P})$.  By~\Cref{expression:valid-fixture:1}
  defined in Section~\ref{ssec:terms:properties},
  $\SOtwo=\bigcup H(\calF_{PBT})$.  We get that $\closedHS
  \subseteq \bigcup H(\calF_{BT}$). In other words,
  $H(\calF_{BT})$ is a covering set of $\closedHS$.  By
  Lemma~\ref{lemma:C}, there exists a subset $\calR \subset
  H(\calF_{BT})$, such that (i) $\calR$ is a covering set, and (ii)
  $|\calR| \in \{3,4,5\}$.  We prove separately for $|\calR| \in\{3,4\}$ and
  $|\calR| = 5$.

  If $|\calR| \in\{3,4\}$, there exist $i \in \{3,4\}$ hemispheres that
  correspond to $i$ base facets of $i$ bodies and fingertips,
  respectively, of at most four fingers, which we choose as the
  fingers of $G'$.

  If $|\calR| = 5$, then $\calR$ contains an open hemisphere
  $\openHS_t$, such that $\openHS_t = h(f_t)$ and the base facet of
  the palm and $f_t$ are parallel.\footnote{Similar conditions are
    described in the proof of Lemma~\ref{lemma:C}.} In
  a polyhedron, two parallel facets cannot be neighbors; thus, $f_t$
  must be the base facet of a fingertip of some finger $F$. Let $f_b$
  denote the base facet of the body of the finger $F$ and set
  $\openHS_b = h(f_b)$. Observe, that $\calR_1 = \calR \setminus
  \{\openHS_t\}$ must be a covering set of the unit circle
  $\boundaryHS_t$, and $|\calR_1| = 4$. Observe that $\boundaryHS_b
  \neq \boundaryHS_t$; thus, $\calR_2 = \calR_1 \setminus
  \{\openHS_b\}$ is a covering set of a closed semicircle $\closedCS$
  and $|\calR_2| = 3$. Following a deduction similar to the above,
  there exist three hemispheres that correspond to three base facets
  of three bodies or fingertips, respectively, of at most three fingers,
  which we choose as the fingers of $G'$ in addition to $F$.

  A polyhedron that admits the lower bound is depicted in
  Figure~\ref{fig:special}. Proving that a snapping fixture for this
  polyhedron with less then four fingers does not exists is done using
  our generator. We exhaustively searched the configurations space and
  did not find a valid snapping fixture.
  \qed\end{proof}

\noindent%
Proof of Observation~\ref{observation:one-finger}
\begin{proof}
  Let $G$ be a fixture with only one finger. Then,
  $|H(\calF_{PBT})| = 3$.  However, by Lemma~\ref{lemma:D} the
  minimum size of a covering set of $\SOtwo$ is four.
\end{proof}

\section{Limitations and Future Research}
\label{sec:future}

\subsection{Form Closure}
\label{ssec:future:form-closure}
\setlength{\intextsep}{-2pt}%
\setlength{\columnsep}{0pt}%
\begin{wrapfigure}[5]{R}{3.5cm}
  \includegraphics[width=3.5cm]{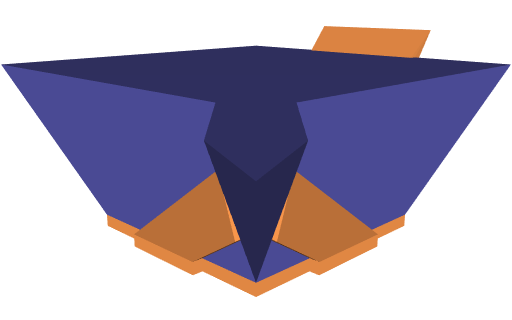}
\end{wrapfigure}
Our generator synthesizes fixtures that do not necessarily prevent
angular motion. Such fixtures are rarely obtained. Nevertheless, the
figure to the right depicts a workpiece and a snapping fixture
(synthesized by our generator), such that the workpiece can escape the
assembled configuration using torque.  However, other snapping fixture
of this workpiece that guarantee form-closure of the workpiece do
exist (and offered by our generator). Devising efficient synthesis
algorithms for guaranteeing form closure is left for future research.%
\setlength\intextsep{\intextsepSaved}%
\setlength\columnsep{\columnsepSaved}%

\subsection{Spreading Degree}
\label{ssec:future:spreading-degree}
Increasing the \emph{spreading degree} (see
Section~\ref{sssec:terms:spreading-degree}) will enable the synthesis
of fixtures for a larger range of workpieces.  Future research could
result with (i) a classification of polyhedra according to the minimal
spreading degree required for their snapping fixtures, and (ii)
algorithms for synthesis of fixtures with a larger fixed spreading
degree or even unlimited.

\subsection{Joint Flexibility}
\label{ssec:future:flexibility}
The flexibility of the \emph{joints} is an important consideration in
the design. In order to construct a snapping fixture, the joint that
connects the body of a finger to the palm, and the joint that connects
the fingertip of a finger to its body must allow the rotation of the
respective parts about the respective axes when force is applied. Some
of the subtleties of this flexibility are discussed below. For
simplicity we move the discussion to the plane, where our workpiece
and snapping fixture are polygons.

\setlength{\intextsep}{-3pt}%
\setlength{\columnsep}{0pt}%
\begin{wrapfigure}[10]{r}{3.8cm}
  \centering%
  \begin{tikzpicture}
    \tikzstyle{myarrows}=[line width=1mm,draw=blue,-triangle 45,postaction={draw, line width=3mm, shorten >=4mm, -}]
    \begin{scope}[shift={(0,-1.1)}]
      \coordinate (p1) at (-0.5,-0.513);
      \coordinate (p2) at (1,0);
      \coordinate (p3) at (1,2);
      \coordinate (p4) at (-0.5,2);
      \draw[pattern color=color2!50,pattern=crosshatch dots] (p1)--(p2)--(p3)--(p4) to[out=-60,in=120] cycle;
    \end{scope}
    \coordinate (p1) at (-0.5,-0.513);
    \coordinate (p2) at (1,0);
    \coordinate (p3) at (1,2);
    \coordinate (p4) at (-0.5,2);
    \draw[fill=color2] (p1)--(p2)-- node[left]{$a$} (p3)--(p4) to[out=-60,in=120] cycle;
    \draw [myarrows](0.25,-0.75)--(0.25,0.75);
    \coordinate (g1) at (-0.5,2);
    \coordinate (g2) at (1,0);
    \coordinate (g3) at (-0.5,0);
    \draw[fill=color1] (g1)-- ++(1.5,0) node(axis){}-- ++([rotate=30]0,-2)-- node[above]{$b$} ++([rotate=50]-1.53,0)-- ++([rotate=20]1,0)-- ++([rotate=50]1,0)-- ++([rotate=30]0,2.6) node(x){}-- ++(-1.8,0) to[out=-60,in=120] cycle;
    \centerarc[red,very thick,dotted](1,2)(270:300:1.1)
    \node [red] at (1.2,1.2) {$\theta$};
    \draw[>=latex,xshift=1cm,yshift=2cm,<->] ([rotate=30]0.15,0)-- node[right]{$a$} ++([rotate=30]0,-2);
    \centerarc[red,very thick,dotted](1,0)(90:200:0.45)
    \node [red] at (0.82,0.16) {$\eta$};
    \centerarc[red,very thick,dotted](1.99,0.27)(120:230:0.45)
    \node [red] at (1.75,0.28) {$\eta$};
    \node [color1!50!black] at (0.3,2.24) {Palm};
    \node (as) at (0.0,1.7) {Axis};
    \draw [thick,-{Stealth[scale=1.5]}] (as)to[out=0,in=210](axis);
    \node[point] at (axis){};
    \draw[dashed] (axis) -- (x);
  \end{tikzpicture}
\end{wrapfigure}
Let's focus on one finger. Consider the configuration where the finger
is about to snap.  Assume, for further simplicity, that the joint that
connects the body and the fingertip of the finger is rigid, and consider
only the joint that connects the finger with the palm, as depicted in
the figure to the right. This configuration occurs a split second
before the assembly reaches the assembled state when translated,
starting at the serving configuration. Let $\theta$ denote the angle
between the finger and the workpiece. Note that in the assembled
configuration $\theta$ equals $0$ for all fingers. Let $\theta_{c}$
denote the joint threshold angle, that is, the maximum bending angle
the finger can tolerate without breaking. The threshold angle of every
joint depends on the material and thickness of the region around the
joint. $\theta$ is an angle of a triangle with one edge lying on the
body base-facet and another edge lying on the fingertip base-facet. Let
$a$ and $b$ denote the lengths of these edges, respectively, and let
$\eta$ be the angle between them. The finger will break when $\theta >
\theta_{c}$.  Applying the law of sines, we get $b = \frac{a \sin
  \theta }{\sin (\pi-\theta-\eta)} = \frac{a \sin \theta
}{\sin(\theta+\eta)}$, which implies a maximal value $b \leq
\frac{\min (a) \sin \theta }{\sin(\theta+\eta)}$.  On the other hand,
the characteristics of the material of the finger determine the
minimal value of $b$ that guarantees a secured grasp of the workpiece
by the fingertip. The construction of a fixture $G$ is feasible, only if
selecting a proper value $b$ for every finger of $G$ is possible. We
remark that the full analysis in space is more involved, and for now
our generator does not take into account material properties such as
flexibility.
\setlength\columnsep{\columnsepSaved}%

\subsection{Gripping Strength}
\label{ssec:future:strength}
Another consideration in the fixture design is the gripping strength.
The gripping strength of a finger is based on the angle between the
palm and the body of the finger and on the angle between the body of
the finger and the fingertip of the finger. The gripping strength is in
opposite relation with these angles; that is, the smaller each one of
these angles is the stronger the gripping is. While our generator
currently does not take in account strength considerations, it could
be used as a criterion in ranking valid snapping fixtures of a given
workpiece.

\section{Assortment of Interesting Workpieces and Fixtures}
\label{sec:assortment}
\begin{figure}[!ht]
  \centering%
  \subfloat[]{\label{fig:sub_tet:1}
    \includegraphics[width=2.5cm]{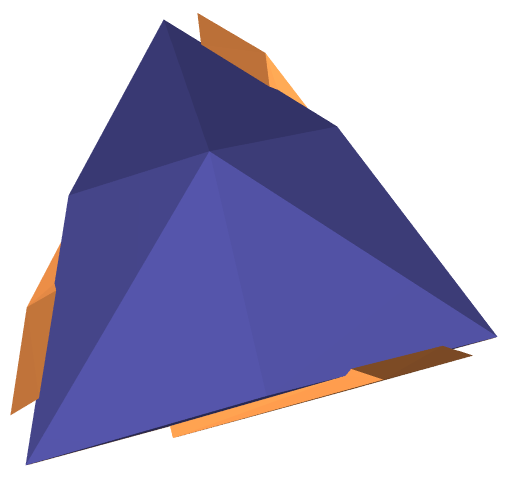}}\quad
  \subfloat[]{\label{fig:sub_tet:2}
    \includegraphics[width=2.5cm]{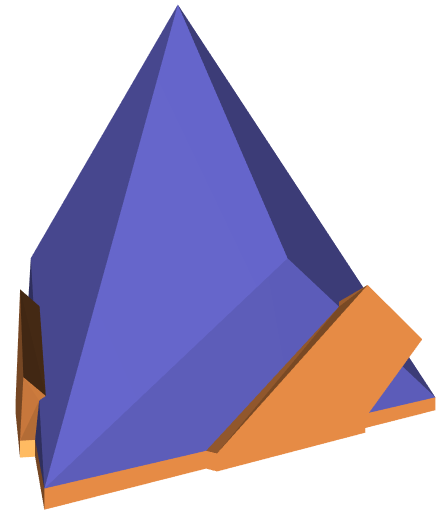}}\quad
  \subfloat[]{\label{fig:sub_tet:3}
    \includegraphics[width=2.5cm]{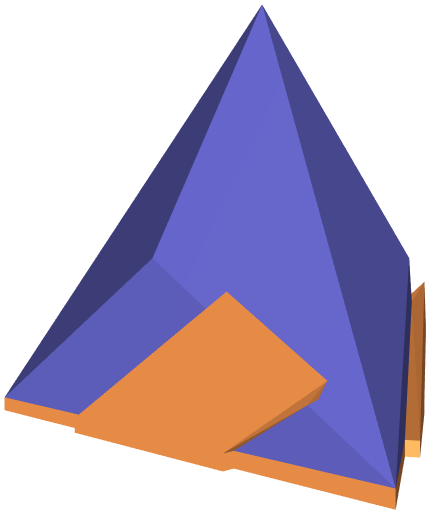}}
  \setlength{\abovecaptionskip}{\abovecaptionskipSaved}
  \setlength{\belowcaptionskip}{10pt}
  \caption[]{%
  (\subref*{fig:sub_tet:1}),(\subref*{fig:sub_tet:2}),(\subref*{fig:sub_tet:3})
  Three different views of a polyhedron with~10 facets and a three-finger
  snapping fixture.}\label{fig:R3}
\end{figure}

Figures~\ref{fig:sub_tet:1}, \ref{fig:sub_tet:2},
and~\ref{fig:sub_tet:3} depict a polyhedron $P$ and a fixture with
three fingers that snaps onto $P$. They demonstrate case I in
Section~\ref{sec:algorithm}. Here, we fix the base facet of the palm.
It holds that for every possible fixture of $P$ with the fixed palm in
the figures $|\calR|=3$.  To construct the polyhedron $P$ in the
figures we start with a regular tetrahedron (such as the one depicted
in Figure~\ref{fig:tet}), fix the bottom facet, subdivide each one of
the remaining three facets into three identical triangles, and
slightly translate the newly introduced vertex in the direction of the
outer normal to the original facet, ensuring that the dihedral angle
between the bottom facet and its neighbor remains acute.

\begin{figure}[!ht]
  \centering%
  \subfloat[]{\label{fig:par_3_3_1-1}
    \includegraphics[width=2.5cm]{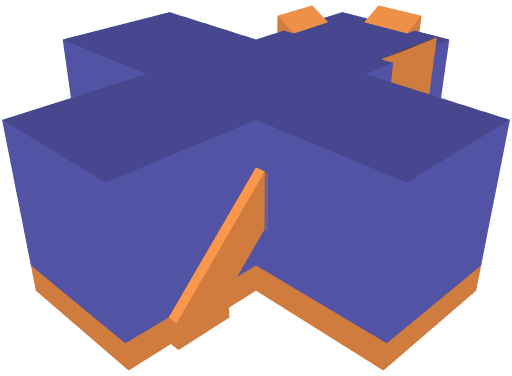}}\quad
  \subfloat[]{\label{fig:par_3_3_1-2}
    \includegraphics[width=2.5cm]{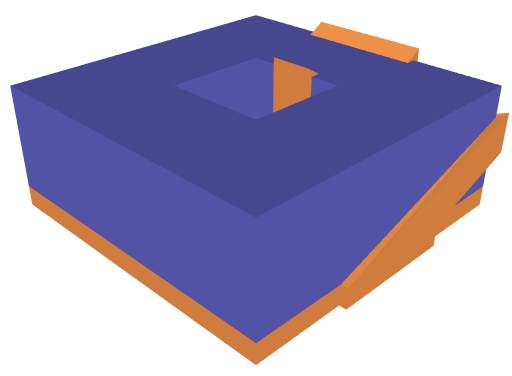}}\quad
  \subfloat[]{\label{fig:no-fixture}
    \includegraphics[width=2.5cm]{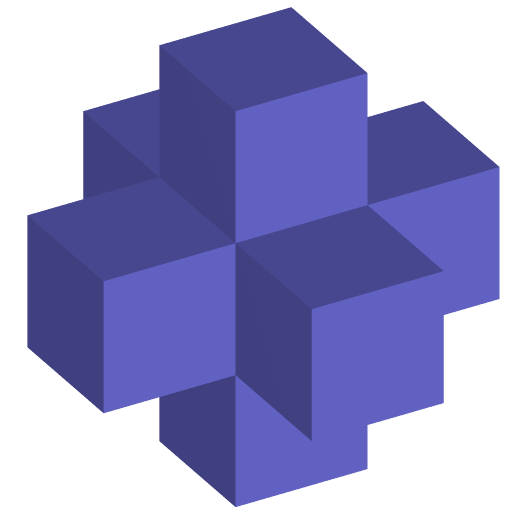}}
  \setlength{\abovecaptionskip}{\abovecaptionskipSaved}
  \setlength{\belowcaptionskip}{10pt}
  \caption[]{%
    (\subref*{fig:par_3_3_1-1}),(\subref*{fig:par_3_3_1-2}) Two polyhedra
    and their snapping fixtures, respectively.
    (\subref*{fig:no-fixture}) A polyhedron that does not have a valid
    snapping fixture.}\label{fig:R5}
\end{figure}

Figures~\ref{fig:par_3_3_1-1} and~\ref{fig:par_3_3_1-2} depict two
polyhedra, $P_1$ and $P_2$, and their snapping fixtures,
respectively. They demonstrate case I in
Section~\ref{sec:algorithm}. The number of facets of each polyhedron
is larger then six; however, it holds that for every possible fixture
of $P_i$, $|\calR|=5$, where $\calR \subset H(\calF_{BT})$ and $\calR$
is a covering set of the closed hemisphere $\SOtwo\setminus
H(\calF_{P})$.

There exists a polyhedron $P$ that does not have a valid fixture and
the cardinality of the minimal covering set of $H(\calF^P)$ is~6,
where $\calF^P$ is the set of all facets of the polyhedron $P$; see
the Figure~\ref{fig:no-fixture}.

\begin{figure}[!ht]
  \vspace{0pt}%
  \centering%
  \subfloat[]{\label{fig:octa1}
    \includegraphics[width=2.5cm]{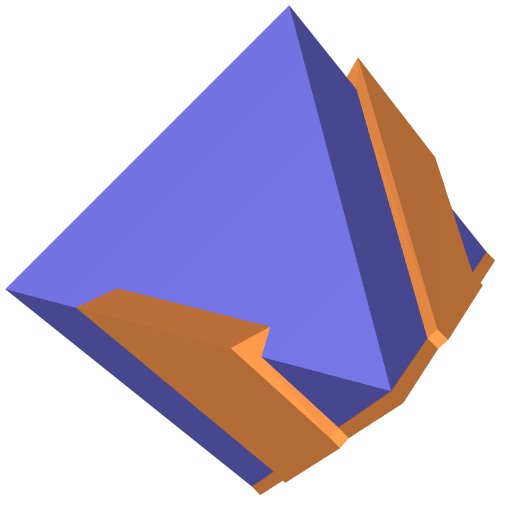}}
  \subfloat[]{\label{fig:octa2}
    \includegraphics[width=2.5cm]{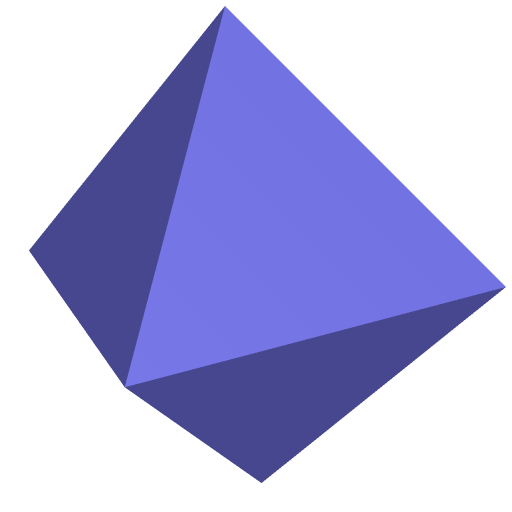}}
  \subfloat[]{\label{fig:octagonal-pyramid}
    \includegraphics[width=2.5cm]{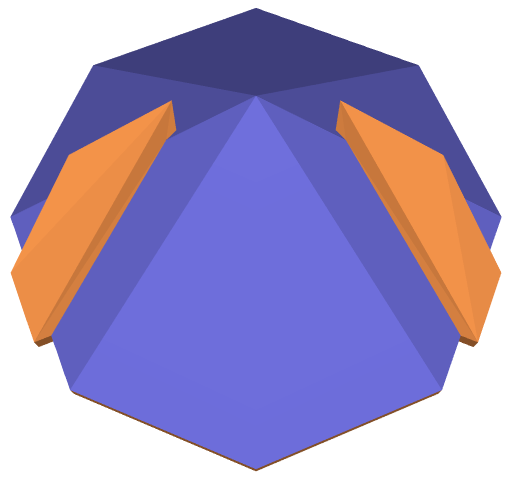}}\quad
  \setlength{\abovecaptionskip}{\abovecaptionskipSaved}
  \setlength{\belowcaptionskip}{10pt}
  \caption[]{%
    (\subref*{fig:octa1}),(\subref*{fig:octa2}) Tow different views of
    An octahedron and a three-finger snapping fixture.
    (\subref*{fig:octagonal-pyramid}) An octagonal-pyramid and a two-finger
    snapping fixture.
  }
  \vspace{0pt}%
\end{figure}

Figures~\ref{fig:octa1} and~\ref{fig:octa2} depict an octahedron and a
snapping fixture with three fingers, which is the minimum in this
case.  Figure~\ref{fig:octagonal-pyramid} depict an octagonal pyramid
and a snapping fixture with two fingers.

\end{document}